\documentclass[a4paper,12pt]{article}
\usepackage[bbgreekl]{mathbbol}
\usepackage[small]{titlesec}
\usepackage[utf8]{inputenc}
\usepackage[T1]{fontenc} 
\usepackage[]{mdframed}
\usepackage{lineno}
\usepackage[backend=biber,style=alphabetic,natbib=false,giveninits=true,doi=true,isbn=false,url=false,date=year,maxbibnames=99,sorting=nyt,defernumbers=true]{biblatex}
\DeclareNameAlias{default}{family-given}

\DeclareFieldFormat[article]{title}{#1}
\renewbibmacro{in:}{}
\DeclareFieldFormat[article]{pages}{#1}
\DeclareFieldFormat{journal}{#1}
\renewcommand\bf\bfseries

\renewbibmacro*{volume+number+eid}{%
	\printfield{volume}%
	\setunit*{\addnbspace}
	\printfield{number}%
	\setunit{\addcomma\space}%
	\printfield{eid}}
\DeclareFieldFormat[article]{number}{\mkbibparens{#1}}
\DeclareFieldFormat[article]{volume}{\textbf{#1}}
\DeclareFieldFormat{year}{\mkbibparens{#1}}
\DeclareBibliographyDriver{article}{%
	\printnames{author}:%
	\newunit\newblock
	\newunit\newblock
    \printfield[titlecase]{title}
	\printfield{journaltitle}
	\newunit
	\iffieldundef{number}{\printfield{volume}}{\printfield{volume}\addspace\printfield{number}}
	\addcomma\addspace\printfield{pages}\addspace
	\printfield{year}
}
\AtEveryBibitem{\clearfield{month}}
\AtEveryBibitem{\clearfield{day}}
\usepackage[english]{babel}
\usepackage[T1]{fontenc}
\usepackage{csquotes}
\usepackage{bbm}
\usepackage[leqno]{amsmath}
\usepackage{esint}
\usepackage{amsfonts,amsthm,amsbsy,amssymb,dsfont,stmaryrd}
\usepackage[dvipsnames]{xcolor}
\usepackage{braket}

\makeatletter
\newcommand{\leqnomode}{\tagsleft@true\let\veqno\@@leqno}
\newcommand{\reqnomode}{\tagsleft@false\let\veqno\@@eqno}
\makeatother

\usepackage{mathtools}
\usepackage[makeroom]{cancel}

\numberwithin{equation}{section}

\newcommand\myshade{85}
\colorlet{mylinkcolor}{violet}
\colorlet{mycitecolor}{YellowOrange}
\colorlet{myurlcolor}{Aquamarine}

\usepackage[left=2.5cm,right=2.5cm,top=2.5cm,bottom=2.5cm]{geometry}
\usepackage[unicode=true,pdfusetitle,
bookmarks=true,bookmarksnumbered=false,bookmarksopen=false,
breaklinks=false,pdfborder={0 0 1},backref=false,
linkcolor = mylinkcolor!\myshade!black,
citecolor = mycitecolor!\myshade!black,
urlcolor  = myurlcolor!\myshade!black,
colorlinks = true,
]
{hyperref}
\usepackage[nameinlink]{cleveref}

\usepackage{tikz}
\usetikzlibrary{arrows}
\usetikzlibrary{arrows.meta}
\usetikzlibrary{intersections}
\usetikzlibrary{calc}
\usepackage{graphicx}
\usepackage{caption}
\usepackage{slashed}
\definecolor{ct_black}{HTML}{000000}
\definecolor{ct_orange}{HTML}{ED872D}
\definecolor{ct_purple}{HTML}{7A68A6}
\definecolor{ct_blue}{HTML}{348ABD}
\definecolor{ct_turquoise}{HTML}{188487}
\definecolor{ct_red}{HTML}{E32636}
\definecolor{ct_pink}{HTML}{CF4457}
\definecolor{ct_green}{HTML}{467821}

\definecolor{ct2_green}{HTML}{9FF781}
\definecolor{ct2_green_dark}{HTML}{088A08}

\theoremstyle{plain}
\newtheorem{thm}{\protect\theoremname}[section]
\theoremstyle{plain}
\newtheorem{lem}[thm]{\protect\lemmaname}
\theoremstyle{plain}

\theoremstyle{plain}
\newtheorem{prop}[thm]{\protect\propositionname}
\theoremstyle{plain}
\newtheorem{claim}[thm]{\protect\claimname}
\newtheorem{assumption}[thm]{\protect\assumptionname}

\theoremstyle{remark}
\newtheorem{rem}[thm]{\protect\remarkname}

\theoremstyle{definition}
\newtheorem{defn}[thm]{\protect\definitionname}
\theoremstyle{plain}

\providecommand{\assumptionname}{Assumption}
\providecommand{\conventionname}{Convention}
\providecommand{\claimname}{Claim}
\providecommand{\corollaryname}{Corollary}
\providecommand{\definitionname}{Definition}
\providecommand{\lemmaname}{Lemma}
\providecommand{\propositionname}{Proposition}
\providecommand{\remarkname}{Remark}
\providecommand{\theoremname}{Theorem}
\providecommand{\examplename}{Example}

\crefname{section}{Section}{Sections}
\crefname{appendix}{Appendix}{Appendices}
\crefname{figure}{Figure}{Figures}
\crefname{assumption}{Assumption}{Assumptions}
\crefname{thm}{Theorem}{Theorems}
\crefname{lem}{Lemma}{Lemmas}
\crefname{rem}{Remark}{Remarks}
\crefformat{equation}{(#2#1#3)}
\crefname{table}{Table}{Tables}

\crefrangelabelformat{equation}{(#3#1#4--#5#2#6)}

\crefmultiformat{equation}{(#2#1#3}{, #2#1#3)}{#2#1#3}{#2#1#3}
\Crefmultiformat{equation}{(#2#1#3}{, #2#1#3)}{#2#1#3}{#2#1#3}

\newtheorem*{lem*}{\protect\lemmaname}

\newcommand{\ee}{\operatorname{e}}
\newcommand{\ii}{\operatorname{i}}

\newcommand{\vzero}{v^{\radial}}

\newcommand{\NN}{\mathbb{N}}
\newcommand{\RR}{\mathbb{R}}
\newcommand{\CC}{\mathbb{C}}

\newcommand{\calF}{\mathcal{F}}

\newcommand{\calO}{\mathcal{O}}
\newcommand{\calI}{\mathfrak{P}}

\newcommand{\calP}{\mathcal{P}}

\newcommand\norm[1]{\left\lVert#1\right\rVert}

\newcommand\abs[1]{\left|#1\right|}

\newcommand{\ip}[2]{\langle #1, #2 \rangle}

\newcommand{\dif}{\operatorname{d}}

\renewcommand{\Re}[1]{\operatorname{\mathbb{R}\mathbbm{e}}\left\{#1\right\}}

\newcommand{\vf}{\varphi}
\newcommand{\Id}{\mathds{1}}

\newcommand{\dist}{\mathrm{dist}}

\newcommand{\supp}{\operatorname{supp}}

\newcommand{\radial}{\circ}

\newcommand{\nc}{\normalcolor}

\usepackage{environ}

\NewEnviron{malign}{%
	\begin{align}\begin{split}
			\BODY
	\end{split}\end{align}
}

\newcommand{\eq}[1]{\begin{align*}#1\end{align*}}
\newcommand{\eql}[1]{\begin{align}#1\end{align}}

\newcommand{\br}[1]{\left(#1\right)}

\title{Magnetic double-wells: Absence of Tunneling}
\author{\href{cf@math.princeton.edu}{Charles L. Fefferman}, \href{mailto:jacobshapiro@princeton.edu}{Jacob Shapiro}\\
	{\footnotesize Department of Mathematics, Princeton University}\\
	 \href{miw2103@columbia.edu}{Michael I. Weinstein}\\
		\footnotesize{Department of Applied Physics and Applied Mathematics,}\\\footnotesize{and Department of Mathematics, Columbia University}
}

\begin{document}
	\reqnomode
	
	\maketitle
	\begin{abstract}
		We present a magnetic double-well Hamiltonian where the tunneling between the two wells vanishes, as recently announced in \cite{FSW24}.
	\end{abstract}

\section{Introduction}

Quantum tunneling in double-well systems is a cornerstone phenomenon in mathematical physics. In a non-magnetic setting, a particle localized in one of two deep and well-separated potential wells has a nonzero probability of tunneling to the other well. The tunneling time is inversely proportional to the distance between the lowest two eigenvalues of the double-well Hamiltonian, a quantity
which is (typically) exponentially small in the well-depth and well-separation.
This paradigm, developed both heuristically and rigorously, goes back to the classical literature and underlies a great variety of effects in physics and chemistry; see, for example, \cite{Landau_Lifshitz_vol_3} in physics or \cite{Simon_1984_10.2307/2007072,Helffer_Sjostrand_1984} in mathematics. The latter two works derive a lower bound on the eigenvalue splitting assuming the single-well potential has a single non-degenerate minimum. In \cite{FLW17_doi:10.1002/cpa.21735} lower bounds are derived assuming rather that the single-well is compactly supported.

In contrast, adding a \emph{constant magnetic field} fundamentally alters the picture. Early works on the magnetic case resolved the case of weak magnetic fields, confirming agreement with the non-magnetic case \cite{Helffer_Sjostrand_1987_magnetic_ASNSP_1987_4_14_4_625_0}. For strong magnetic fields, magnetic translations endow localized states with nontrivial complex phases, yielding interference among tunneling paths. For example in \cite{FSW_22_doi:10.1137/21M1429412} we showed that for double-well magnetic systems whose single-well is radial, the magnetic field effect leads to an exponentially smaller tunneling amplitude. However, such systems with radial single-well potential were shown to still possess a strictly positive lower bound on tunneling \cite{FSW_22_doi:10.1137/21M1429412,HelfferKachmar2024,Morin2024}.

Recently we have announced \cite{FSW24} that the magnetic setting admits a phenomenon that has no analogue in the non-magnetic case: \emph{exact vanishing of tunneling} in a double-well system built from a suitable non-radial single-well potential. In this paper we present a complete proof. We construct a family of (non-radial) single-well potentials such that, in a strong constant magnetic field and for large well depth, the associated symmetric double-well has zero eigenvalue splitting between its two lowest levels, so that tunneling is entirely eliminated. Moreover, by varying parameters within this family one can flip the parity of the ground state (from even to odd), 
via a coalescence and re-emergence of two distinct eigenvalues. 

A recent preprint \cite{exner2025fleamagneticelephant} studies diminished tunneling due to asymmetry of the single-well Hamiltonian in a setting with ``magnetic wells'' and no electronic potential.

Let us make the above discussion more precise by setting up notation. We work with the Hilbert space $L^2(\RR^2)$ and on it consider magnetic Schr\"odinger operators of the form
\eq{
  (P-\frac12\lambda X^\perp)^2+V_\lambda(X)
} where $P\equiv-\ii\nabla$ is the momentum operator and $X$ is the position operator, with \eql{X^\perp\equiv(-X_2,X_1)\,;} $\lambda>0$ is a sufficiently large coupling constant which controls simultaneously the scaling of the magnetic field strength (which grows as $\lambda$) and the depth (and possibly shape) of the potential $V_\lambda:\RR^2\to(-\infty,0]$. Such Hamiltonians describe the dynamics of an electron bound to a two-dimensional plane and subject to a constant magnetic field pointing in the negative $3$-axis direction. We consider simultaneously two main operators of this type: the \emph{single-well} Hamiltonian $h_\lambda$ and \emph{double-well} Hamiltonian $H_\lambda$. They are defined as 
\eql{
  h_\lambda := (P-\frac12\lambda X^\perp)^2+v_\lambda(X)
} and 
\eql{\label{eq:double-well Hamiltonian}
  H_\lambda := (P-\frac12\lambda X^\perp)^2+v_\lambda(X+d)+v_\lambda(-X+d)\,.
} Here $v_\lambda:\RR^2\to(-\infty,0]$ is a smooth compactly supported (say within $B_a(0)$ for some $a>0$) single-well potential and $d\equiv(d_1,0)\in\RR^2\setminus\Set{0}$ is the displacement of each well taken, without loss of generality, to lie along the $1$-axis. The notation $v_\lambda$ indicates that $v_\lambda$ depends on $\lambda$. We emphasize that $x\mapsto v_\lambda(x+d)+v_\lambda(-x+d)$ is \emph{inversion} symmetric, and contrast it with the commonly studied case of $x\mapsto v_\lambda(x+d)+v_\lambda(x-d)$. The two cases agree if $v_\lambda$ itself is inversion symmetric, i.e. $v_\lambda(-x)=v_\lambda(x)$ for all $x\in\RR^2$. We generally allow for $\norm{v_\lambda}_\infty$ to have order of magnitude $\lambda^2$.

In the simplest scenario, e.g., that considered in \cite{FSW_22_doi:10.1137/21M1429412} we made the choice $v_\lambda := \lambda^2 v$ but here we shall be concerned with more general $\lambda-$ dependence of $v_\lambda$. Indeed, we shall be interested in potentials of the form
\eq{
  v_\lambda := \lambda^2 v^{\radial} + \tau_{\lambda,M,D}\sum_{\nu=1}^4 W_{\lambda,\nu}\ ,
} where $v^{\radial}:\RR^2\to[-1,0]$ is a smooth \emph{radial} potential which obeys all the hypotheses of \cite{FSW_22_doi:10.1137/21M1429412} (see \cref{sec:single-well-v0-h0} below), $\tau_{\lambda,M,D}$ is a strictly positive small parameter to be chosen depending on $\lambda$, and $W_{\lambda,\nu}$ are small perturbing potentials, whose shape also depends on $\lambda$.

Thanks to our assumptions on $v^{\radial}$ and the assumption that $\lambda$ is sufficiently large, $h_\lambda$ has ground state energy $e_{0,\lambda}=-\lambda^2+\mathrm{o}(\lambda^2)$ which has distance to the remainder of the spectrum with at least order of magnitude $c_{\mathrm{gap}}$, an order one constant. It follows that $E_{0,\lambda}\leq E_{1,\lambda}$, the two lowest eigenvalues of $H_\lambda$, are also at a distance of at least order $c_{\rm gap}$ from the remainder of the spectrum of $H_\lambda$. We are chiefly concerned with the \emph{double-well eigenvalue splitting} defined by \eq{
  \Delta_0(\lambda) := E_{1,\lambda} - E_{0,\lambda}\,.
} 

The Hamiltonian $H_\lambda$ defined in \cref{eq:double-well Hamiltonian} enjoys an inversion symmetry $[H_\lambda,\calI]=0$ where $\calI:\psi(x)\mapsto\psi\br{-x}$; not necessarily so for $h_\lambda$. As a result, we may decompose $L^2(\RR^2)$ according to the even and odd eigenspaces of $\calI$ \eql{\label{eq:Hilbert space decomposition into even and odd}
L^2(\RR^2) =: L^2_{\mathrm{even}}\oplus L^2_{\mathrm{odd}}
}
and $H_\lambda$ is diagonal w.r.t. that decomposition. We denote the ground state energy of $\left.H_\lambda\right|_{L^2_\sigma}$ for $\sigma=\mathrm{even},\mathrm{odd}$ by $E_{0,\sigma,\lambda}$. Thus a further quantity of interest is the \emph{signed} eigenvalue difference, \eql{\label{eq:signed eigenvalue difference}
\mathfrak{S}_0(\lambda) := E_{0,\mathrm{odd},\lambda}-E_{0,\mathrm{even},\lambda}\,.
} 
 We emphasize that $\mathfrak{S}_0(\lambda)$ need not be positive. If $\lambda$ is sufficiently large we shall see that \eq{
  \Set{E_{0,\lambda}, E_{1,\lambda}} = \Set{E_{0,\mathrm{odd},\lambda},E_{0,\mathrm{even},\lambda}}
} (\cref{prop:even-odd-eigs}) and the question is whether $E_{0,\lambda}=E_{0,\mathrm{even},\lambda}$ or $E_{0,\lambda}=E_{0,\mathrm{odd},\lambda}$. An emergent quantity here is the \emph{hopping coefficient} which will also play a role in the present analysis. Its definition (\cref{def:hopping})  is given by
\eql{\label{eq:hopping coefficient}
  \rho_0(\lambda) &:= \left\langle\widehat{R}^{-d}\vf_{0,\lambda}\br{H_\lambda-e_{0,\lambda}\Id}\calI\widehat{R}^{-d}\vf_{0,\lambda}\right\rangle\\
  &= \int_{x\in\RR^2} \overline{\vf_{0,\lambda}(x+d)}v_\lambda(x+d) \exp\br{\ii\lambda d_1 x_2}\vf_{0,\lambda}(-x+d)\dif{x}\nonumber
} Here $\vf_{0,\lambda}$ is the $L^2-$ normalized ground state of $h_\lambda$; $\widehat{R}^{ d }$ is Zak's magnetic translation operator \cite{Zak_1964_PhysRev.134.A1602}, given by \eql{\label{eq:magnetic translations}\widehat{R}^z := \exp\br{-\ii z\cdot \br{P+\frac12\lambda X^\perp}}\qquad(z\in\RR^2)\,.} Since $z\cdot x^\perp = -z^\perp\cdot x$,   \eql{
 (\widehat{R}^{ z }f)(x) \equiv \exp\br{\ii\frac\lambda2 x \cdot  z^\perp} f(x-z)\qquad(x,z\in\RR^2\,,f:\RR^2\to\CC)\,.
}

In the non-magnetic setting, one always has $\mathfrak{S}_0(\lambda)>0$ and $\rho_0(\lambda)<0$. In the present magnetic setting, however, as we shall see, $\rho_0(\lambda)$ remains real (thanks to the $\calI$ symmetry of $H_\lambda$), but can be positive, negative or zero due to an oscillating exponential in \cref{eq:hopping coefficient}. Thus, we now separately have 
\eq{
  \mathfrak{S}_0(\lambda) \approx -2\rho_0(\lambda)
} and \eq{
  \Delta_0(\lambda) = \abs{\mathfrak{S}_0(\lambda)}\,;
} see \cref{prop:Epm1-expand-Delta-expand} below.

Our main theorem is
\begin{thm}\label{thm:main theorem}
  Fix $v^{\radial}:\RR^2\to[-1,0]$, a radial single-well potential which obeys all the hypotheses of \cite{FSW_22_doi:10.1137/21M1429412} (spelled out below in \cref{sec:single-well-v0-h0}). 
  Let $d_1>0$ be given and sufficiently large. Then there exists some $\lambda_\star>0$ (depending on both $v^{\radial},d_1$) such that for all $\lambda\geq\lambda_\star$ there exists a family $\calF_\lambda$ of compactly supported smooth potentials $v_\lambda:\RR^2\to(-\infty,0]$ (dependent on $\lambda$) such that
  \begin{enumerate}
    \item There exists some $ C<\infty$ independent of $\lambda$, such that \eq{
      \sup_{v_\lambda\in\calF_\lambda}\left[\norm{\lambda^2v^{\radial}-v_\lambda}_{L^\infty(\RR^2)} + \mathrm{measure}\br{\Set{x\in\RR^2 | \lambda^2 v^{\radial}(x)\neq v_\lambda(x)}}\right] \leq \exp\br{-C\lambda} \,. 
    } 
    \item \emph{Absence of quantum tunneling}: There exists $v_\lambda\in\calF_\lambda$ such that $\Delta_0(\lambda)=0$. Consequently $H_\lambda$ has no tunneling from its left well to its right well at low energies.
    \item \emph{Vanishing of quantum hopping}: There exists $v_\lambda\in\calF_\lambda$ such that $\rho_0(\lambda)=0$.
    \item \emph{Ground state can take on even or odd parity}: By deforming within the family $\calF_\lambda$, the signed eigenvalue difference, $\mathfrak{S}_0(\lambda)$, can be made to change its sign.
Correspondingly, the ground state eigenfunction $\Phi_{0,\lambda}$ of $H_\lambda$ transitions between 
being symmetric and being anti-symmetric with respect to the origin. At the transition point, the ground state eigenspace of $H_\lambda$ is two-dimensional, spanned by one even and one odd eigenfunction.
  \end{enumerate}
\end{thm}

In forthcoming papers we shall present the remainder of the results announced and conjectured in \cite{FSW24}. In particular, we plan to present
\begin{enumerate}
  \item A lower bound on the \emph{average} splitting, i.e., an upper bound on $-\log\br{\Delta_0}$ as well as hopping $-\log\br{\rho_0}$, when one averages over $\lambda$ in the interval, say $[\Lambda,2\Lambda]$ for $\Lambda>0$ sufficiently large. This shows that zero tunneling does not occur for generic potentials and $\lambda$. However, we will show in an upcoming paper that there is an open set of potentials $v$ in any reasonable topology such that there exists a $\lambda$ at which tunneling vanishes for the one-well potential $\lambda^2v$.
  \item An example of a construction for a double-well potential $V_\lambda:\RR^2\to(-\infty,0]$ for which $\rho_0(\lambda)=0$, and such that $\calI V_\lambda = V_\lambda$ \emph{need not hold}.
  \item An example of a potential $v:\RR^2\to(-\infty,0]$ such that the signed eigenvalue splitting $\mathfrak{S}_0$ and hopping coefficient $\rho_0$, arising from the double-well system associated to the single-well potential $\lambda^2v$ and magnetic field strength $\lambda$, change sign infinitely often as $\lambda\to\infty$.
  
\end{enumerate}

\paragraph{Intuitive explanation of the mechanism.}
Engineered small asymmetric deviations from radial symmetry in a single-well, as a mechanism for reducing and even completely eliminating 
tunneling between wells in quantum magnetic systems, was first explained - in brief - in \cite{FSW24}. We start with what we call a "planet": a radially-symmetric single-well potential $v^{\radial}$ whose tunneling effect is well-understood (since \cite{FSW_22_doi:10.1137/21M1429412}). On top of it we add so-called "sophons": perturbations $W$ which are placed nearby to $v^{\radial}$; their support is exponentially small and moreover they have a coupling constant which is itself also exponentially small. However, the key is where they are placed: thanks to the oscillations afforded by the magnetic translations, we can position the sophons so that: (1) some of the sophon-planet interaction terms dominate the planet-planet interaction term in a basic formula for $\rho_0$ and (2) the overall sign of (the real quantity) $\rho_0$ can be made positive or negative by slightly perturbing the position of the sophons (see \cref{eq:rho_cos} below). Since $\rho_0$ asymptotically controls $\mathfrak{S}_0$ as well as $\Delta_0$, this allows for the vanishing of both.


\paragraph{Consequences and outlook.}
Beyond its intrinsic spectral interest, a robust handle on the sign and magnitude of $\mathfrak{S}_0$ has implications for engineered quantum systems. In periodic settings, suppressed nearest-neighbor hopping leads to dispersionless (flat) bands near the atomic energy, a regime associated with enhanced interaction, nonlinearities, and correlation effects. 
Our construction suggests a possible route to nearly flat magnetic bands in 2D tight-binding models obtained from continuum magnetic Schr\"odinger operators, with quite general (not necessarily translation invariant) underlying atomic lattice, by strong-binding reduction; see the discussion in \cite{ShapWein22}. It is also natural to ask about extensions to three dimensions, other magnetic field orientations, and the interplay with additional (e.g.\ spin or spin--orbit) couplings.

\paragraph{Organization of this paper.} In \cref{sec:results-roadmap} we sketch the roadmap for the proof of \cref{thm:main theorem} and divide it into several main steps. In \cref{sec:implement_the_roadmap} we implement this roadmap and provide any necessary estimates. The appendices recall basic facts about the resolvent of the Landau Hamiltonian, as well as Gaussian decay of magnetic bound states.

\paragraph{Notation and conventions.}
\begin{enumerate}
  \item We set $\hbar=2m=1$ throughout.
  \item The magnetic momentum shall be denoted by 
  \eql{
    \calP_\lambda := P -\frac12\lambda X^\perp\,.
  } Hence $\calP_\lambda^2$ is the Landau Hamiltonian with magnetic field strength $\lambda$ in the symmetric gauge.
  \item Lower case generally refers to single-well objects, e.g. potentials, Hamiltonians, eigenvalues and functions, whereas upper case refers to double-well objects.
  \item Since we are almost exclusively concerned with ground states and their energies it shall be tiresome to carry a subscript zero on all objects, and hence we drop it and assume all states and their associated energies correspond to ground states unless otherwise noted. 
  \item The circle superscript, for  example as in $v^\circ$,  shall refer to the \emph{radial, unperturbed} system, single- or double-well.
    \item Throughout, $x^\perp\equiv(-x_2,x_1)$ and $x\wedge y\equiv x_1y_2-x_2y_1$.
  \item Legend for the list of symbols that shall appear below:
  \begin{enumerate}
\item  $\vf_\lambda$  denotes the $L^2(\RR^2)-$ normalized ground state. Its ground state energy is denoted $e_\lambda$.
     \item $H_\lambda$ is the double-well Hamiltonian (one copy of the single-well \emph{translated} to $-d$, another copy \emph{reflected} and translated to $+d$). Its ground state is denoted $\Phi_\lambda$ and its two lowest energies are $E_{0,\lambda}\le E_{1,\lambda}$. 
     \item $v^{\radial}:\RR^2\to[-1,0]$ denotes the radial single-well potential, the associated  Hamiltonian is $h^{\radial}_{\lambda}\equiv\calP_\lambda^2+\lambda^2v^{\radial}(X)$. The ground state is then $\vf^{\radial}_{\lambda}$ and its ground state energy $e^{\radial}_{\lambda}$. 
     \item Objects appearing with superscript $L$ are \emph{translated} to $-d$. Objects appearing with superscript $R$ are reflected and translated to $d$. Hence $v_\lambda^L$ is the single-well potential centered at $-d$, but $v_\lambda^R$ is the reflected single-well potential, centered at $d$. Moreover, $v^{\radial,L},v^{\radial,R}$ is the \emph{radial} potential $v^{\radial}$ translated to $-d$ or reflected and translated to $+d$ respectively.
     \item $\rho(\lambda)$ is the hopping coefficient, whereas $\rho^\radial$ is the hopping coefficient of the double-well system whose single-wells are radial and given by $v^\radial$.
     \item $\Delta(\lambda)$ is the double-well eigenvalue splitting.
     \item $\mathfrak{S}(\lambda)$ is the double-well signed eigenvalue splitting: it is the difference of the ground state in the odd sector with the ground state in the even sector.     \item $x\mapsto K^\circ(x):=\br{H^{\rm SHO}_\lambda-e^{\radial}_{\lambda}\Id}^{-1}(x,0)$ is the fundamental solution with source at $x=0$ of the harmonic oscillator Hamiltonian:
\eql{
  (P^2+\frac{\lambda^2}{4}X^2-e^{\radial}_{\lambda}\Id)K^\circ\equiv\delta\,.
} Here $\delta$ is the Dirac delta distribution. Likewise, \eq{K:=\br{H^{\rm SHO}_\lambda-e_{\lambda}\Id}^{-1}(\cdot,0)\,.} 
  \end{enumerate}
  \item Various parameters used throughout:
  \begin{enumerate}
        \item $M$ is sufficiently large.
      \item $D$ is sufficiently large dependent on $M$.
      \item $\lambda$ is sufficiently large dependent on $D$.      
      \item $d_1= \frac12D^{3/2}$
      \item $\delta_{\lambda,M,D}=\exp\br{-M \lambda D^{3/2}}$.
      \item $\tau_{\lambda,M,D}=\exp\br{-\frac\lambda4\br{2D^{5/2}-3D^2}}$
      \item $\operatorname{diam}(\supp(v_\lambda))\sim D$.
  \end{enumerate}
\end{enumerate}
\subsection*{Acknowledgements}
 CF was supported in part by NSF grant DMS-1700180. MIW was supported in part by NSF grants DMS-1908657, DMS-1937254, DMS-2510769  and Simons Foundation Math + X Investigator Award \# 376319 (MIW). Part of this research was carried out during the 2023-24 academic year, when MIW was a Visiting Member in the School of Mathematics - Institute of Advanced Study, Princeton, supported by the Charles Simonyi Endowment, and a Visiting Fellow in the Department of Mathematics at Princeton University.
The authors wish to thank Antonio Cordoba and David Huse for stimulating discussions. We thank P.A. Deift and J. Lu for their careful reading of the \cite{FSW24} manuscript, and insightful comments and questions. 
 
\section{The proof of \cref{thm:main theorem}}\label{sec:results-roadmap}

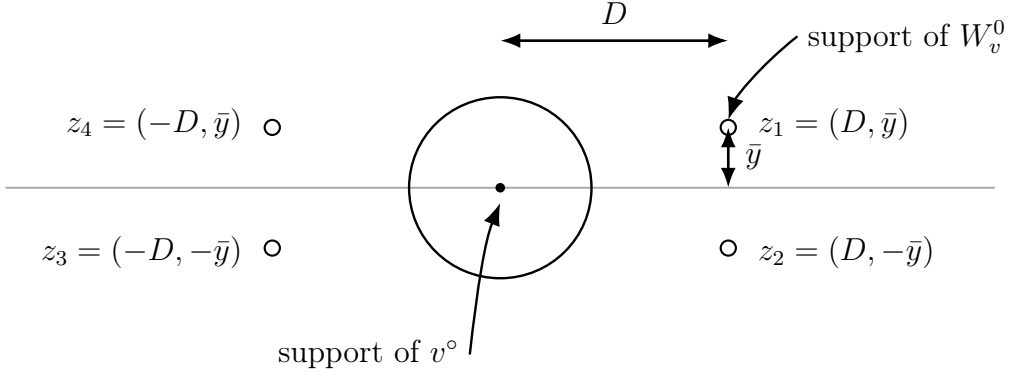
\begin{figure}
 \begin{center}
   \begin{tikzpicture}[>=Latex, line cap=round, line join=round]
 \def\R{1.2}     
 \def\D{3.0}     
 \def\ybar{0.8}    
 \def\axisHalf{6.5}  
 \def\ringsize{0.10} 

 \draw[gray!70, line width=0.7pt] (-\axisHalf,0) -- (\axisHalf,0);

 \draw[line width=0.9pt] (0,0) circle (\R);
 \fill (0,0) circle (1.8pt);

 \foreach \sx/\sy in { \D/\ybar, \D/-\ybar, -\D/\ybar, -\D/-\ybar }{
  \draw[line width=0.8pt] (\sx,\sy) circle (\ringsize);
 }

 \node[anchor=west] at (\D+0.25, \ybar+0.05) {$z_{1}=(D,\bar y)$};
 \node[anchor=west] at (\D+0.25,-\ybar-0.05) {$z_{2}=(D,-\bar y)$};
 \node[anchor=east] at (-\D-0.25,-\ybar-0.05) {$z_{3}=(-D,-\bar y)$};
 \node[anchor=east] at (-\D-0.25, \ybar+0.05) {$z_{4}=(-D,\bar y)$};

 \draw[<->, line width=1.0pt]
  (0, \R+0.75) -- (\D, \R+0.75)
   node[midway, above=2pt] {$D$};

 \draw[<->, line width=1.0pt]
  (\D, 0) -- (\D, \ybar)
   node[midway, right=2pt] {$\bar y$};

 \draw[-{Latex[length=3mm]}, line width=0.9pt]
  (-0.4,-2.2) .. controls (-0.3,-1.3) and (-0.2,-0.7) .. (0,-0.2);
 \node[anchor=east] at (-0.4,-2.2) {$\text{support of }v^{\circ}$};

 \draw[-{Latex[length=3mm]}, line width=0.9pt]
  (3.9,2.0) .. controls (3.4,1.6) and (3.1,1.2) .. (\D,\ybar+\ringsize);
 \node[anchor=west] at (3.9,2.0) {$\text{support of }W_{v}^{0}$};
\end{tikzpicture}
 \end{center}
  \caption{Support of the  “atomic well”, $v_\lambda$. 
One radially symmetric planet with four ``sophons''; see \cref{sec:single-well-plus-sophons}.}
  \label{fig:schematic_single-well}
\end{figure}

We first provide a detailed roadmap for the proof of \Cref{thm:main theorem} and then, in \Cref{sec:implement_the_roadmap}, we implement the steps of the roadmap, providing the technical arguments.

\subsection{Roadmap, Step 1:
The single-well radial potential, $\lambda^2\vzero(x)$, the magnetic atomic Hamiltonian, $h^{\radial}_{\lambda}$, and its ground state $\vf^{\radial}_{\lambda}(x)$}\label{sec:single-well-v0-h0}

Fix a smooth potential $\vzero:\RR^2\to\RR$.
Assume constants $C_0, C_1, C_2,c_{\rm gap}$ and $c_2 $, depending on $\vzero$, such that
\begin{itemize}
  \item[(V1)] $-C_0\le \vzero\le0$
  \item[(V2)] $\supp \vzero\subset B_1(0)$
  \item[(V3)] \emph{Radial ground state assumption}: Introduce \eql{
 h^{\radial}_{\lambda}\equiv\  \calP_\lambda^2+\lambda^2 \vzero(X)\ , \label{eq:mag-atomic}
 }
 the magnetic atomic Hamiltonian with radial potential $\vzero$.
 \begin{assumption}[Radial ground state of $h^{\radial}_{\lambda}$]\label{assume:radial-phi0}
 For all $\lambda\ge C_1$, the ground state eigenfunction $\vf^{\radial}_{\lambda}$ of $h^{\radial}_{\lambda}$ is radial, i.e. $\vf^{\radial}_{\lambda}(x)=\vf^{\radial}_{\lambda}(y)$ if $\norm{x}=\norm{y}$; we may choose $\vf^{\radial}_{\lambda}\geq0$.
 \end{assumption}

\item[(V4)] The ground state eigenvalue $e^{\radial}_{\lambda}$ of $h^{\radial}_{\lambda}$ 
  satisfies
\eql{\label{eq:gs-large_lam} 
-C_2\lambda^2 \le e^{\radial}_{\lambda}\le -c_2\lambda^2,\quad \textrm{for $\lambda\ge C_1$}.
}
\item[(V5)] \emph{Spectral gap assumption}:
\begin{assumption}[Spectral gap of $h^{\radial}_{\lambda}$]\label{assume:gap}
 %
 %
 $h_0^\lambda$ has a spectral gap which is uniform in $\lambda$. Specifically, there is a constant $c_{\rm gap}>0$, such that for all $\lambda\ge C_1$
 \begin{equation}\label{eq:spectral-gap} {\rm dist}\Big(e^{\radial}_{\lambda},\sigma(h^{\radial}_{\lambda})\setminus\{e^{\radial}_{\lambda}\}\Big)\ge c_{\rm gap}.\end{equation}
\end{assumption}
\end{itemize}

\begin{rem}
  The assumptions above may be verified for potentials $v^{\radial}$ which have a unique non-degenerate minimum so that the approximation by the quantum harmonic oscillator holds, see e.g. \cite{Matsumoto_1994} for \cref {assume:gap} and \cite{HelfferKachmar2024} for \cref{assume:radial-phi0}.
\end{rem}

We adopt a convention where constants: $c, C, C^\prime$, etc. are determined by $C_0,C_1,c_2,C_2,c_{\rm gap}$ in (V1)-(V5). These symbols may denote different constants in different occurrences. For two functions $f$ and $g\geq0$ of $\lambda$, we write $f=\calO(g)$ to denote $|f(\lambda)|\leq C g(\lambda)$ for all $\lambda\geq\lambda_{\rm min}$.

\begin{prop}\label{eq:phi0-K} Let $\vf^{\radial}_{\lambda}$ denote the (radial) ground state of $h^{\radial}_{\lambda}$. Then, 
\[ \vf^{\radial}_{\lambda}(x) = \Gamma K^\circ(x),\quad \norm{x}>1,\] 
where \eql{\label{eq:SHO f solution}
x\mapsto K^\circ(x):=\br{H^{\rm SHO}_\lambda-e^{\radial}_{\lambda}\Id}^{-1}(x,0)} is the fundamental solution with source at $x=0$ of the harmonic oscillator Hamiltonian:
\eql{
  (P^2+\frac{\lambda^2}{4}X^2-e^{\radial}_{\lambda}\Id)K^\radial\equiv\delta
} where $\delta$ is the Dirac delta distribution and $\Gamma\equiv\Gamma_\lambda$ is a positive constant satisfying upper and lower bounds:
\eql{ \exp\left(-C\lambda\right) \ \le\ \Gamma\ \le\ \exp\left(C\lambda\right)\,.
\label{eq:Gamma-bounds}
}
\end{prop}
\begin{proof}[Proof of \Cref{eq:phi0-K}]
  In the region $r=\norm{x}>1$, outside of the support of $\vzero$, $\vf^{\radial}_{\lambda}$ and $K^\radial$, considered as functions of the radial coordinate, both satisfy the equation: 
  \[ -\partial_r^2 f(r) - \frac{1}{r} \partial_r f(r) + \frac{\lambda^2}{4}r^2f(r) = e^{\radial}_{\lambda} f(r),\]
  which has a one dimensional subspace of solutions tending to zero as $r\to\infty$, spanned by $K^\radial$. It follows that $\vf^{\radial}_{\lambda}=\Gamma K^\radial$ for some constant $\Gamma$ which is positive since both functions $\vf^{\radial}_{\lambda}$ and $K^\radial$ are positive.

  To bound $\Gamma$, pick some $z\in\RR^2$ with $\norm{z}=:C>1$. Then by \cite[Theorem 2.3]{FSW_22_doi:10.1137/21M1429412}, we have
  \eq{
    \exp(-C_z\lambda)\leq \vf^{\radial}_{\lambda}(z) \leq \exp(-c_z\lambda)
  } for two constants $C_z,c_z\in(0,\infty)$ independent of $\lambda$, if $\lambda$ is sufficiently large. The estimates derived in \cref{lem:bounds on the SHO Green's function} (below in the appendix) imply the same for $K^\radial$:
  \eq{
    \exp(-\widetilde{C_z}\lambda)\leq K^\radial(z) \leq \exp(-\widetilde{c_z}\lambda)
  } for two constants $\widetilde{C_z},\widetilde{c_z}\in(0,\infty)$ independent of $\lambda$, if $\lambda$ is sufficiently large. Invoking $\Gamma = \frac{\vf^{\radial}_{\lambda}(z)}{K^\radial(z)}$ we then find
  \eq{
  \exp\br{-\br{C_z-\widetilde{c_z}}\lambda}\leq \Gamma \leq \exp\br{\br{\widetilde{C_z}-c_z}\lambda}\,.
  } 
\end{proof}
\bigskip

\begin{prop}[Gaussian decay of $\vf^{\radial}_{\lambda}$]\label{prop:gs0-decay} There exists some $C<\infty$ independent of $\lambda$ such that for all $x\in\RR^2\setminus\supp(v^\circ)$ we have
\eql{\label{eq:gs0-decay1a}
|\vf^{\radial}_{\lambda}(x)|
\le C\lambda^2\ \exp\left(-\frac{\lambda}{4}[{\rm dist}(x,\supp(\vzero))]^2\right)\,.
}
\end{prop}
For the proof, see \cref{prop:Gaussian decay of bound states for compactly supported potentials}.

\subsection{Constants and parameters}\label{sec:parameters} We shall require constants $M, D_{\rm min}, \lambda_{\rm min}$ to be chosen later such that 
\begin{itemize}
  \item[(M)] $M$ is a large enough constant determined by $C_0, C_1, c_2, C_2, c_{\rm gap}$ in (V1)-(V5) above
  \item[(D)] $D\ge D_{\rm min}$ is a large enough constant determined by $C_0, C_1, c_2, C_2, c_{\rm gap}$ and $M$ in (V1)-(V5) and (M) above.
  \item[($\bar{\rm Y}$)] 
  $\bar{y}\in[0,1]$.
  \item[($\lambda_{\rm min}$)] $\lambda_{\rm min}$ is a large enough constant determined by $C_0, C_1, c_2, C_2, c_{\rm gap}, M$ and $D$ in (V1)-(V5), (M) and (D) above.
\end{itemize}

We let 
\eql{ \lambda\ge \lambda_{\rm min}.
\label{eq:lam_ge_lam_min}
}
Further, we introduce parameters
$\delta$ and $\tau$, which specify a single-well configuration; see the schematic in \Cref{fig:schematic_single-well}:
\begin{subequations}
\label{eq:delta&tau}
\begin{align} 
\delta &= \delta_{\lambda,M,D}\equiv \exp\left(-M\lambda D^{3/2}\right)\label{eq:delta}\\
\tau_{\lambda,M,D}&= \tau_{\lambda,M,D} \equiv 
\exp\left(-\frac{\lambda}{4
}(2D^{5/2}-3D^2)\right) \ .\label{eq:tau}\
\end{align}
\end{subequations}
The choice of $\delta$ is explained in \Cref{sec:Omega0nu0-expand}, and that of $\tau$ in \Cref{sec:tau-determined}.

\subsection{Roadmap, Step 2: The radial single-well potential perturbed by ``sophons''}\label{sec:single-well-plus-sophons}

We shall perturb the single-well radial potential $\lambda^2\vzero(x)$ by disjointly supported potentials with small amplitude and support. 
With the above definitions, \Cref{eq:delta&tau}, of $\delta=\delta_{\lambda,M,D}$ and $\tau=\tau_{\lambda,M,D}$ in \Cref{eq:delta&tau}, we next build up the {\it perturbed single-well potential}.

Fix a \emph{radial} smooth function $W_{0,\lambda}:\RR^2\to\RR$ satisfying:
\begin{itemize}
  \item[(S1)] $-1\le W_{0,\lambda}\le0$
  \item[(S2)] ${\rm supp}\ W_{0,\lambda} \subset B_{\delta}(0)$
  \item[(S3)] $\fint_{B_\delta(0)} W_{0,\lambda} \equiv \frac{1}{\pi\delta^2}\int_{B_\delta(0)}W_{0,\lambda} = -c$, for some $c>0$.
\end{itemize}

A finite sum of translates of $W_\lambda$ will be added to $\lambda^2\vzero$. These will be centered at points $\{\zeta_\nu\}_{\nu=1,\cdots,\nu_{\rm max}}\subset \RR^2$, with 
\begin{itemize}
  \item[(S4)] $\Big|\ \norm{\zeta_\nu}-D\ \Big|\le1,\quad \nu=1,\dots, \nu_{\rm max}$, where 
\item[(S5)] $\nu_{\rm max}\le C$. Below we will take $\nu_{\rm max}\le 4$; see \Cref{fig:schematic_single-well}.
\end{itemize}
The {\it perturbed single-well potential}, which will serve as our single-well potential of our
 double-well construction, is given by:
\eq{
v_\lambda = \lambda^2 \vzero + \tau_{\lambda,M,D}W,\quad \big( W=W_\lambda \big),
}
where 
\eql{\label{eq:W-sophons}
 W =  \sum_{\nu=1}^{\nu_{\rm max}} W^0_\nu,\quad W^0_\nu(x) = W_{0,\lambda}(x-\zeta_\nu)
}
 Hence, $W= W_\lambda$ is supported in $\Set{x:\norm{x}\ge D-3}$. We emphasize that we allow the perturbing sophon potential, $W(x)$, to depend on $\lambda$. Note also that for a general collection of centers, $\{\zeta_\nu\}$, $W(x)\ne W(-x)$ and hence 
 $v_\lambda(-x)\ne v_\lambda(x)$. The double-well potential, $V_\lambda(x)$, constructed below will have even symmetry with respect to $x=0$, and for our particular choice of centers $\zeta_\nu$ below, $v_\lambda$ will in fact have inversion symmetry too.

 \paragraph{Terminology} We shall refer to 
$\lambda^2 \vzero(x)$ as ``the planet''
 and the terms $\tau_{\lambda,M,D}W_\nu^0(\cdot)$ as ``the sophons''. By \eqref{eq:tau}, (S1)-(S3), the sophons are a very small perturbation of the planet, e.g. their $C^{100}$ norm can be taken to be exponentially small in $\lambda\gg1$. Nevertheless, they have large effect
 \cite{liu2014three}. 
 
\begin{rem}
    The superscript zero, as in  $W_\nu^0$ in \cref{eq:W-sophons}, indicates that this collection of "sophons" is centered about the origin. This is in contrast to the left and right translates below, which are indicated with superscripts $L,R$.
\end{rem}

The {\it perturbed single-well Hamiltonian} or {\it single-well Hamiltonian with sophons},  is the operator:
 \eq{
    h_\lambda= \calP_\lambda^2 + v_\lambda(X) = h^{\radial}_{\lambda} + \tau_{\lambda,M,D}W_\lambda(X),
 }
 with $\tau=\tau(\lambda)$ as in \Cref{eq:delta&tau}.
Its ground state eigenpair is denoted $(e_\lambda,\vf_\lambda)$, where $e_\lambda<0$:
\begin{align}\label{eq:gs-single-well}
 h_\lambda\ \vf_\lambda(x) = e_\lambda\ \vf_\lambda(x),\quad \vf_\lambda\in L^2(\RR^2).
\end{align}
The eigenpair $(e_\lambda,\vf_\lambda)$ is a small perturbation of the ground state eigenpair $(e^{\radial}_{\lambda},\vf^{\radial}_{\lambda})$ of the unperturbed Hamiltonian $h^{\radial}_{\lambda}$. The following Gaussian decay bound holds for 
 $\vf_\lambda(x)$ as well; see \cref{prop:Gaussian decay of bound states for compactly supported potentials}.
\begin{prop}[Gaussian decay of $\vf_\lambda$]\label{prop:Gauss-bound-phi_lambda}
The ground state of $ h_\lambda$ satisfies a Gaussian decay bound with the identical rate as the corresponding bound
for $\vf^{\radial}_{\lambda}$.\\
There exist some constant $C\in(0,\infty)$ independent of $\lambda$ such that for all $x\in\RR^2\setminus\supp(v_\lambda)$,
\eql{
|\vf_\lambda(x)|
\le C\lambda^2\ \exp\left(-\frac{\lambda}{4}\Bigg[{\rm dist}\big(x,\supp\ v_\lambda \big)\ \Bigg]^2\right)\,.
\label{eq:gs-decay1a}}
\end{prop}

\subsection{Road map, Step 3: The double-well Hamiltonian, $H_{\lambda}$}\label{sec;gs_perturbed_onewell}

We construct our double-well potential, which is inversion symmetric with respect to $x=0$. First left-translate $ v_\lambda$ (\Cref{sec:single-well-plus-sophons}), and center it at
\eql{ - d := \Big(-\frac{D^{3/2}}{2},0\Big)\ . \label{eq:xi-def}
}
Call this translated single-well potential $v_\lambda^L $.
To $v_\lambda^L $ we add $v_\lambda^R:=\calI v_\lambda^L$, its inversion with respect to $x=0$ to obtain $V_\lambda$, the double-well potential.
 
In explicit detail, we introduce the left well
\begin{align}
 v_\lambda^L(x) &\equiv v_\lambda(x+ d ) \label{eq:V_L} \\
 & =\lambda^2 \vzero(x+ d ) + \tau\ \sum_{\nu=1}^{\nu_{\rm max}} W^0\big(x-(\zeta_\nu - d )\big) \nonumber\\
 & \equiv\lambda^2 v^{\radial,L} + \tau\ \sum_{\nu=1}^{\nu_{\rm max}} W_\nu^L(x) \nonumber\\
 &\equiv \lambda^2 v^{\radial,L} + \tau_{\lambda,M,D}W^L(x).
 \nonumber \end{align}
 To build up the right well,
 first we take 
 \[ v^{\radial,R}(x) \equiv v^{\radial,L}(-x).\]
 Note $v^{\radial,R}(x)=\vzero(-x+ d )=\vzero(x- d )$, a right translate of $\vzero(x)$, since $\vzero$ is radially symmetric. Next we define
 \eql{
 W_\nu^R(x) \equiv W_\nu^L(-x)\quad {\rm and}\quad 
 W^R(x) \equiv \sum_{\nu=1}^{\nu_{\rm max}} W_\nu^R(x),
 \label{eq:WRLnu}}
 and finally
 \eql{\label{eq:V_R}
   v^R_\lambda(x) \equiv v^L_\lambda(-x)\ =\ \lambda^2 v^{\radial,R}(x) + \tau_{\lambda,M,D}W^R(x)\,.
 }

 %
 %
 %
%
The double-well potential, $V_\lambda$, is the sum of the left and right single-well potentials:
\eql{ V_\lambda = v_\lambda^L + v^R_\lambda.
\label{eq:W2well}
}
By construction
\[ V_\lambda(-x) = V_\lambda(x).\]
\Cref{fig:schematic_2WELL} is a schematic of the support of $ V_\lambda$; two radially symmetric ``planets'', $\lambda^2 v^{\radial,L}$ and $\lambda^2 v^{\radial,R}$, each surrounded by four symmetrically placed ``sophons'',
 $W^L_\nu$ and $W^R_\nu$, where $\nu=1,2,3,4$.
 In this example, $ v_\lambda(-x)= v_\lambda(x)$ and hence
 $v_\lambda^R(x)= v_\lambda(x- d )$.

 \begin{rem}
     A minimal example of the phenomenon presently described would employ only \emph{one} sophon, so that $v_\lambda^R$ is merely an inversion of $v_\lambda^L$. However, here we have elected to use four sophons precisely to guarantee that $v_\lambda^R$ is a translate (and not merely an inversion) of $v_\lambda^L$. This is a more realistic setup when one considers a whole lattice of wells.
 \end{rem}

\begin{figure}[h!]
 \begin{center}

 \begin{tikzpicture}[>=latex, line cap=round, line join=round]
 \def\R{1.1}     
 \def\Sep{7.0}    
 \def\gap{0.8}    
 \def\h{0.28}    


 \draw[gray!70, line width=0.6pt] (-2,0) -- (\Sep/2-\gap/2,0);
 \draw[gray!70, line width=0.6pt] (\Sep/2+\gap/2,0) -- (\Sep+2,0);

 \foreach \x in {0,\Sep} {
  \draw[black, line width=0.8pt] (\x,0) circle (\R);
  \fill (\x,0) circle (1.6pt);
 }

 \foreach \c/\x in {L/0,R/\Sep} {
  \foreach \sx in {-1,1} {
   \draw[line width=0.7pt] ({\x+\sx*(\R+0.55)}, 0.55) circle (0.09);
   \draw[line width=0.7pt] ({\x+\sx*(\R+0.55)},-0.55) circle (0.09);
  }
 }

 \draw[<->, line width=0.9pt] (-2,0) -- (-2,{0.55})
  node[midway,left=3pt] {$\tilde y$};
 \draw[<->, line width=0.9pt] (0,{\R+0.35}) -- ({0+(\R+0.55)},{\R+0.35})
  node[midway,above=2pt] {$D$};
 \draw[<->, line width=1.0pt] (-0.2,-1.8) -- (\Sep+0.2,-1.8)
  node[midway,below=2pt] {$D^{3/2}$};
\end{tikzpicture}
\end{center}
  \caption{Double-well configuration consisting of two atomic wells (\cref{fig:schematic_single-well}) centered at the points $x=\pm d= (\pm D^{3/2}/2,0)$.}
  \label{fig:schematic_2WELL}
\end{figure}
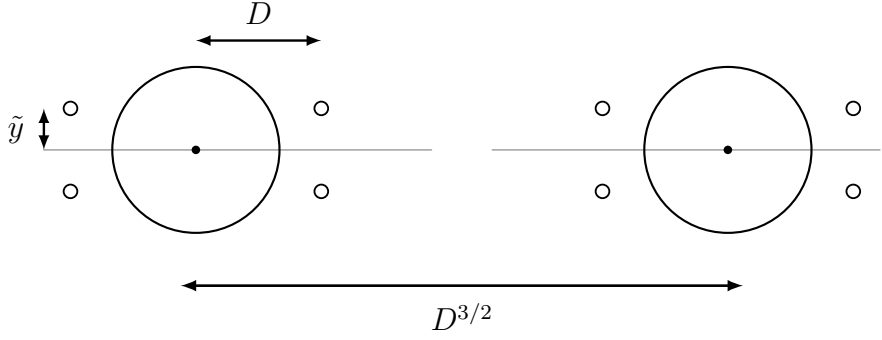
 
 By \cref{eq:xi-def}, the distance between the centers of two atomic (single-well) configurations is:
 \eql{ 2\norm{d} \equiv D^{3/2}. \label{eq:xiD}}
 By taking $D\ge D_{\rm min}\gg1$ in hypothesis (D) above, we arrange that this distance is much larger than the horizontal distance between the sophons 
 and their respective planets:
 \eql{ D^{3/2}\gg D. \label{eq:xi_vs_D}}

Due to our symmetric placement of the sophons, the right well is a translate of the left well.

The double-well Hamiltonian is given by 
\eq{
 H_\lambda = 
 \calP^2_\lambda
+ V_\lambda, 
\label{eq:H2well}
}
where $ V_\lambda$ is defined in \Cref{eq:V_L}, \Cref{eq:V_R}, \Cref{eq:W2well}.

\begin{prop}\label{prop:H2well-low_energy}
The lowest part of the spectrum of $ H_\lambda $ consists of two real eigenvalues $E_{0,\lambda}\le E_{1,\lambda}$ in a neighborhood of the single-well ground state energy $e_\lambda$ and 
 \eq{\dist\br{\Set{E_0,E_1}, \sigma\br{ H_\lambda }\setminus\Set{E_0,E_1}}\ge \frac{1}{2}c_{\rm gap} \ .
 }
\end{prop}
The proof of this proposition may be found, e.g., in \cite[Theorem 6.3]{fefferman2025lowerboundsquantumtunneling}

Introduce the left and right atomic Hamiltonians: $ h_\lambda^L$ and $ h_\lambda^R$ and their ground state eigenpairs:
\[ (e_\lambda,\vf_\lambda^L)\quad\textrm{and}\quad\ ( e_\lambda,\vf_\lambda^R),\quad \textrm{which satisfy}\]

\begin{subequations}
\label{eq:LRwells}
\begin{align}
 h_\lambda^L\ \vf_\lambda^L &\equiv\br{ \calP_\lambda^2 + v_\lambda^L(X)} \vf_\lambda^L = e_\lambda\ \vf_\lambda^L,
\label{eq:left-well-evp}\\
 h_\lambda^R\ \vf_\lambda^R &\equiv \br{\calP_\lambda^2 + v^R_\lambda(X)} \vf_\lambda^R = e_\lambda\ \vf_\lambda^R\ .
\label{eq:right-well-evp}
\end{align}
\end{subequations}

Since $ v^R_\lambda (x)= v^L_\lambda (-x)$, we have (after possibly multiplying the 
ground states by constant factors of absolute value equal to one):
\eql{\label{eq:phiRphiL}
  \vf_\lambda^R(x)=\vf_\lambda^L(-x).
}
The states $\vf_\lambda^L$ and $\vf_\lambda^R$ can be expressed in terms of $\vf_\lambda$, the ground state of $ h_\lambda$, whose potential is centered at $x=0$; see \cref{eq:gs-single-well}. Using the magnetic translation operator introduced in \cref{eq:magnetic translations} we have: 
\eql{ \vf_\lambda^L(x) = \widehat{R}^{- d }\vf_\lambda(x) =
 \ee^{-\ii\frac{\lambda}{2}x\cdot d^\perp }\vf_\lambda(x+ d ) = \ee^{\ii\frac{\lambda}{2}x\wedge d }\vf_\lambda(x+ d )\label{eq:phiL}
 }
 and, by \Cref{eq:phiRphiL},
 \eql{
 \vf_\lambda^R(x) = 
 \ee^{+\ii\frac{\lambda}{2}x\cdot d^\perp }\vf_\lambda(-x+ d )=\ee^{-\ii\frac{\lambda}{2}x\wedge d}\vf_\lambda(-x+ d )\,.
 \label{eq:phiRR}}
Here $x\wedge y\equiv x_1 y_2-x_2y_1$.
 

\bigskip

By \cref{eq:left-well-evp} we have 
\begin{align}
\br{\calP_\lambda^2 - e_\lambda\Id } \vf_\lambda^L = -v_\lambda^L(X) \vf_\lambda^L . \label{eq:left-well-evp-1}
\end{align}


From \Cref{sec:HO&Landau}, in particular \Cref{eq:relation between Landau and SHO Green's functions}, we have that \eq{
\RR^2\times\RR^2\ni(x,y)\mapsto K(x-y) \ee^{-\ii\frac{\lambda}2 x\wedge y}} is the fundamental solution of $\calP_\lambda^2 - e_\lambda\Id$ with source at $y\in\mathbb R^2$; we use the notation $K$ from \cref{eq:SHO f solution}.
In particular from \Cref{eq:left-well-evp-1} it follows that
 \begin{align}
  \vf_\lambda^L(x) = -\int_{y\in\RR^2} K(x-y)\ \ee^{-\ii\frac{\lambda}2 x\wedge y}\ v_\lambda^L(y)\vf_\lambda^L(y)\dif{y}.
  \label{eq:phi_L_rep}
\end{align}

\subsection{Roadmap, Step 4: The hopping coefficient for the double-well Hamiltonian}

\begin{defn}\label{def:hopping}
  The magnetic hopping coefficient is given by 
  \eql{
  \rho = \rho(\lambda) = \Big\langle \vf_\lambda^L,\br{H_\lambda -e_\lambda\Id}\vf_\lambda^R\Big\rangle,
  \label{eq:hopping_coefficient}
  }
  where the inner product is taken in $L^2(\RR^2)$.
\end{defn}

A consequence of the inversion symmetry of the double-well potential, $ V_\lambda$, is

\begin{prop}\label{prop:rho-real}
\begin{enumerate} 
\item The hopping coefficient $\rho(\lambda)$ is real-valued.
\item 
$ \rho(\lambda)=\left\langle \vf_\lambda^L, v^R_\lambda\vf_\lambda^R\right\rangle
= \left\langle v^L_\lambda\vf_\lambda^L,\vf_\lambda^R\right\rangle$.
\end{enumerate}
\end{prop}
\begin{proof} Note that $ H_\lambda -e_\lambda\Id$ is self-adjoint. Further, $[\calI, H_\lambda ]=0$,
 where $\calI f(x)=f(-x)$ and hence $\calI\vf_\lambda^L(x)=\vf_\lambda^L(-x)=\vf_R(x)$. From these properties we have:
\begin{align*}
\rho &= \Big\langle ( H_\lambda -e_\lambda)\vf_\lambda^L,\vf_\lambda^R\Big\rangle
 = \Big\langle \vf_\lambda^L, ( H_\lambda -e_\lambda)\vf_\lambda^R\Big\rangle\\
 &=\Big\langle \calI\vf_\lambda^L, \calI( H_\lambda -e_\lambda)\vf_\lambda^R\Big\rangle=
 \Big\langle \calI\vf_\lambda^L, ( H_\lambda -e_\lambda) \calI\vf_\lambda^R\Big\rangle\\
 &= \Big\langle \vf_\lambda^R, ( H_\lambda -e_\lambda) \vf_\lambda^L\Big\rangle =
 \overline{\Big\langle ( H_\lambda -e_\lambda) \vf_\lambda^L,\vf_\lambda^R\Big\rangle} = \overline{\rho},
\end{align*}
completing the proof of Part 1. Part 2 follows 
from the relations
\begin{subequations}
\label{eq:HphiLR}
\begin{align}
( H_\lambda -e_\lambda)\vf_\lambda^L(x) &= v^R_\lambda (x)\vf_\lambda^L(x),\\
( H_\lambda -e_\lambda)\vf_\lambda^R(x) &= v^L_\lambda (x)\vf_\lambda^R(x),
\end{align}
\end{subequations}
which are consequences of \cref{eq:LRwells}.
\end{proof}

\begin{rem}[Non-magnetic case]\label{rem:nonmag-rho-neg}
 Note that in the non-magnetic case, $\rho(\lambda)$ is negative. Indeed, then the non-magnetic ground states may be chosen positive whereas $v\leq0$.
\end{rem}

\subsection{Roadmap, Step 5: Decomposition of the double-well hopping coefficient into a sum of interaction terms}
\label{sec:RM-Step5}
We next seek to express the hopping coefficient, \Cref{eq:hopping_coefficient}, in terms of interactions among "planets" and "sophons". Thanks to \Cref{eq:phi_L_rep} and \cref{prop:rho-real} we have
\eql{\label{eq:hopping_rep}
-\rho = \int_{\RR^2} \dif{x} \int_{\RR^2} \dif{y}\ 
 v^R_\lambda (x) \overline{\vf_\lambda^R(x)}\ K(x-y)\ \ee^{-\ii\frac{\lambda}2 x\wedge y}\ v_\lambda^L(y) \vf_\lambda^L(y).
}

The fundamental solution $K(x)$ is displayed in \Cref{sec:HO&Landau}. 

Using now the expressions in \Cref{eq:V_L} and \Cref{eq:V_R} for $ v^L_\lambda (x)$, $ v^R_\lambda (x)$ in \Cref{eq:hopping_rep}
 yields an expansion of the magnetic hopping coefficient:
\eql{\label{eq:rho-expand}
-\rho=:\Omega(0,0) + \sum_{\nu=1}^{\nu_{\rm max}}\Big(\ \Omega(\nu,0)\ +\ \Omega(0,\nu)\ \Big)\ +\ \sum_{\nu,\nu^\prime=1}^{\nu_{\rm max}} \Omega(\nu,\nu^\prime).
}
Each term in \Cref{eq:rho-expand} represents an interaction between a planet or sophon of the left atomic configuration ($P^L$ or $S^L_{\nu^\prime}$) and a planet or sophon of the right atomic configuration ($P^R$ or $S^R_{\nu^\prime}$).
The terms, $\Omega(\nu,\nu^\prime)$ in this expansion are, for $\nu,\nu^\prime=1,\dots,\nu_{\rm max}$\ :
{\footnotesize{
\begin{subequations}
\label{eq:Omegas}
  \begin{align}
&P^R \leftrightarrow P^L: \ \ \Omega(0,0) \equiv \int_{\RR^2\times\RR^2} 
 \lambda^2 v^{\radial,R}(x) 
 \overline{\vf_\lambda^R(x)}\ \ee^{-i\frac{\lambda}2 x\wedge y}\ K(x-y)\  \lambda^2 v^{\radial,L}(y) \vf_\lambda^L(y)\ \dif{x} \dif{y}\qquad\qquad  
 \\
& P^R \leftrightarrow S^L_{\nu}: \ \ \Omega(0,\nu) \equiv \int_{\RR^2\times\RR^2} 
 \lambda^2 v^{\radial,R}(x) \overline{\vf_\lambda^R(x)}\ \ee^{-i\frac{\lambda}2 x\wedge y}\ K(x-y)\ \tau_{\lambda,M,D}W^L_\nu(y) \vf_\lambda^L(y)\ \dif{x} \dif{y} 
 \\
&S^R_{\nu}\leftrightarrow P^L: \ \ \Omega(\nu,0) \equiv \int_{\RR^2\times\RR^2} 
 \tau_{\lambda,M,D}W^R_\nu(x) \overline{\vf_\lambda^R(x)}\ \ee^{-i\frac{\lambda}2 x\wedge y}\ K(x-y)\ \lambda^2 v^{\radial,L}(y) \vf_\lambda^L(y)\ \dif{x} \dif{y} 
 \\
&S^R_{\nu^\prime}\leftrightarrow S^L_\nu: \ \  \Omega(\nu^\prime,\nu) \equiv \int_{\RR^2\times\RR^2} 
 \tau_{\lambda,M,D}W^R_{\nu^\prime}(x) \overline{\vf_\lambda^R(x)}\ \ee^{-i\frac{\lambda}2 x\wedge y}\ K(x-y)\ \tau_{\lambda,M,D}W^L_\nu(y) \vf_\lambda^L(y)\ \dif{x} \dif{y}\ . 
 \end{align}
\end{subequations}
}}

\subsubsection{Strategy}\label{sec:strategy}
 In order to clarify our strategy, we pause to recapitulate and to make a number of observations about the expansion 
\Cref{eq:rho-expand}-\Cref{eq:Omegas}: 
\begin{enumerate}
  \item The wells $v^{\radial,R}$ and $v^{\radial,L}$ (planets) are radially symmetric functions and supported on discs of radius one about $- d $ and $+ d $ respectively
  \item The wells $\tau_{\lambda,M,D}W^L_\nu$ and $\tau_{\lambda,M,D}W^R_\nu$ (sophons), for $\nu=1,\dots,\nu_{\rm max}$, are of depth $\tau$, and are supported on discs of radius $\delta$. 
  \item Each integrand in \cref{eq:Omegas} is supported in a product of discs. 
  \item Our atomic configurations, whose associated double-well potentials have vanishing eigenvalue splitting, will be built by adding to the radial atomic well a few small perturbations of small support. Thus $\delta$ and $\tau$ are chosen to be small, and are governed by the following considerations. 
  \item Note that the three expressions in the expansion of $-\rho$ in \Cref{eq:rho-expand} correspond to:
  \begin{align*}
  &\textrm{$\mathcal{O}(\tau^0)$:\quad Planet $\leftrightarrow $ Planet interaction},\\
  &\textrm{$\mathcal{O}(\tau^1)$:\quad Planet $\leftrightarrow $ Sophon interactions,\ and } \\
   &\textrm{$\mathcal{O}(\tau^2)$:\quad Sophon $\leftrightarrow $ Sophon interactions.}
  \end{align*}
  \item Consider the double-well configuration indicated in the schematic of \Cref{fig:schematic_2WELL}. There are two additional parameters that fix the double-well potential:
   $D$ and $\bar{y}$:\\
$\bullet$\quad $D$ is a fixed large enough number (independent of $\lambda$) which encodes the horizontal distance between planet centers, $2| d |=D^{3/2}$ as well as the horizonal displacement of each sophon relative to its planet, $\pm D$, and \\
 $\bullet$\quad  $\bar{y}$ encodes the vertical displacement, $\pm \bar{y}$, of the sophon centers from the axis through the planet centers.\\
  Hence, the atomic configurations of the type sketched in \Cref{fig:schematic_2WELL} depend on the four parameters:
  \[ \delta,\ \tau, D\quad \textrm{and}\quad \bar{y}.\]
  Here, $\delta=\delta_{\lambda,M,D}$ and 
  $\tau=\tau_{\lambda,M,D}$.
  \item Fix $D_{\rm min}$ sufficiently large. Then, choose $\lambda_{\rm min}$ sufficiently large (possibly depending on $D_{\rm min}$. 
  For any $\lambda\ge\lambda_{\rm min}$, our choices of 
   $\delta_{\lambda,M,D}$ and $\tau_{\lambda,M,D}$ are designated so that
   the dominant term in \Cref{eq:rho-expand} 
   comes from the strongest of the 
   (formally $\mathcal{O}(\tau)$) Planet $\leftrightarrow$ Sophon interactions. Thus, we require
   \begin{subequations}
   \label{eq:PS-PP_SS}
   \begin{align}
    &\Big|\textrm{Planet $\leftrightarrow $ Sophon intrxns}\Big| \gg \Big|\textrm{Planet $\leftrightarrow $ Planet intrxn}\Big| \label{eq:PS-PP}\\
    &\Big|\textrm{Planet $\leftrightarrow $ Sophon intrxns}\Big| \gg \Big|\textrm{Sophon $\leftrightarrow $ Sophon intrxns}\Big|\nonumber \\
  & {\ }\ \label{eq:PS-SS} \end{align}
  \end{subequations}
  The constraint \Cref{eq:PS-PP} imposes a lower bound on $\tau\delta^2$ and 
  the constraint \Cref{eq:PS-SS} imposes
  an upper bound on $\tau\delta^2$.
   These constraints will be shown (below in \Cref{sec:tau-determined} and \Cref{rem:delta-rationale}) to be satisfied 
   with the choices (recall \cref{eq:delta&tau})
   \[ \tau_{\lambda,M,D} =\exp\left(-\frac{\lambda}{4}(2D^{5/2}-3D^2)\right)\quad {\rm and}\quad 
\delta_{\lambda,M,D}= \ee^{-M\lambda D^{3/2}},\]
where $D\ge D_{\rm min}$ and $\lambda\ge\lambda_{\rm min}(D)$.
\end{enumerate}
 
 \bigskip

\subsubsection{Symmetries of interaction terms $\Omega(\nu,\nu^\prime)$}
We conclude this subsection with a remark on the symmetries of $\Omega(\nu,\nu^\prime)$. 
\begin{prop}\label{prop:Omega-nu0-cc-Omega-0nu} Assume $\nu, \nu^\prime\in\{1,\dots,\nu_{\rm max}\}$. Then, 
  $\overline{\Omega(\nu,0)} = \Omega(0,\nu) $. Further, $\Omega(\nu,\nu)$ is real for $\nu=0,1,\dots,\nu_{max}$ and 
  if $\nu\ne\nu^\prime$, then $\overline{\Omega(\nu,\nu^\prime)}= \Omega(\nu^\prime,\nu)$
\end{prop}
\begin{proof}[Proof of \Cref{prop:Omega-nu0-cc-Omega-0nu}] We verify that $\overline{\Omega(\nu,0)} = \Omega(0,\nu) $. The other cases are similar. 

 Let $\tilde x= -y$ and $\tilde y = -x$. Thus, $x\wedge y= \tilde y \wedge \tilde x$ and $x-y=\tilde x - \tilde y$.
 By \cref{eq:WRLnu} and \cref{eq:phiRphiL}:
\begin{align*}
 \overline{\Omega(\nu,0)}
 &= \int_{\RR^2\times\RR^2} 
 \tau_{\lambda,M,D}W^R_\nu(x) \vf_\lambda^R(x)\ \ee^{i\frac{\lambda}2 x\wedge y}\ K(x-y)\ \lambda^2 v^{\radial,L}(y) \overline{\vf_\lambda^L(y)}\ \dif{x} \dif{y}\\
 &= \int_{\RR^2\times\RR^2} 
 \tau_{\lambda,M,D}W^R_\nu(-\tilde y) \vf_\lambda^R(-\tilde y)\ \ee^{i\frac{\lambda}2 \tilde y\wedge \tilde x}\ K(\tilde x- \tilde y)\ \lambda^2 v^{\radial,L}(-\tilde x) \overline{\vf_\lambda^L(-\tilde x)}\ d\tilde x d\tilde y\\
&= \int_{\RR^2\times\RR^2} 
 \tau_{\lambda,M,D}W^L_\nu(\tilde{y}) \vf_\lambda^L(\tilde{y})\ \ee^{-i\frac{\lambda}2 \tilde{x}\wedge\tilde{y}}\ K(\tilde{x}-\tilde{y})\ \lambda^2 v^{\radial,R}(\tilde{x}) \overline{\vf_\lambda^R(\tilde{x})}\ d\tilde{x} d\tilde{y}\\
 &= \int_{\RR^2\times\RR^2} 
 \lambda^2 v^{\radial,R}(\tilde{x}) \overline{\vf_\lambda^R(\tilde{x})} \ee^{-i\frac{\lambda}2 \tilde{x}\wedge\tilde{y}}\ 
 K(\tilde{x}-\tilde{y})\ \tau_{\lambda,M,D}W^L_\nu(\tilde{y}) \vf_\lambda^L(\tilde{y})\ d\tilde{x} d\tilde{y}\ =\ \Omega(0,\nu).
 \end{align*}
  
\end{proof}

\subsection{Roadmap, Step 6: Expansion of the magnetic hopping coefficient, $\rho$}
\label{sec:Roadmap-Step6}
We study the expansion \Cref{eq:rho-expand} for double-well potential, $ V_\lambda$, of the type given schematically in \Cref{fig:schematic_2WELL}, where $\nu_{\rm max}=4$ is the number of sophons surrounding each radially symmetric central well; see \cref{fig:schematic_single-well}.
\begin{prop}[Hopping coefficient expansion]
\label{prop:hopping_leading+error} 
Let $0\le \bar{y}\le 1$, $\lambda_{\rm min}$, $D$, $\delta$ and $\tau_{\lambda,M,D}$ be given as in \cref{sec:parameters}. There exists a positive number ${\rm FAC}_\rho$ satisfying
\eq{
 &  \exp(-CM\lambda D^{3/2}) \le {\rm FAC}_\rho\le \exp(CM\lambda D^{3/2}),\ 
} such that:
  \begin{enumerate}
    \item 
We have

\begin{align*}
  \Big|-\rho-\left[\ (\Omega(1,0)+\Omega(0,1) + \Omega(2,0) + \Omega(0,2)  \right]\ \Big|
  \le \ee^{-\frac{M}{4}\lambda D^{3/2}}\cdot {\rm FAC}_\rho\cdot \ee^{-\frac{1}{4}\lambda \left(D^{3}-D^2)\right)},
\end{align*}
\item The hopping coefficient satisfies the expansion
\begin{align}
  \rho &= -\Bigg(\ 
  \cos\left( \frac{\lambda}{2} D^{3/2}\bar{y} \right) + {\rm ERR}_\rho\ \Bigg)\cdot {\rm FAC}_\rho\cdot 
  \ee^{ -\frac{\lambda}{4}(D^3-D^2) },
\label{eq:rho_cos}\end{align}
where
\begin{align*}
& \Big| {\rm ERR}_\rho \Big| \le C \exp(-\frac{M}{4}\lambda D^{3/2})\,.
\end{align*} 
  \end{enumerate}
\end{prop}
The proof of this proposition is delayed to \cref{sec:from_rho0_2_rho} below.
\begin{thm}[Change in sign of $\rho(\lambda,\bar{y})$, the magnetic hopping]\label{thm:vanishing_hopping}
 Fix $\lambda\ge\lambda_{\rm min}$. Fix parameters: $0\le \bar{y}\le 1$, $\lambda_{\rm min}$, $D$, $\delta$, $\tau_{\lambda,M,D}$ and $M$ as in \cref{sec:parameters}. 
  By varying the vertical displacement parameter, $\bar{y}$, the magnetic hopping coefficient $\rho=\rho(\lambda;\bar{y})$ can be made to take on positive values, negative values or to vanish.
\end{thm} 
\begin{proof}
  Fix constants as above with, in particular, $\lambda_{\min}$ sufficiently large. Now fix $\lambda\ge\lambda_{\rm min}$. 
  Consider the finite set of values $\bar{y}^*_n(\lambda)$, where $n\ge1$, satisfying 
  \[ \frac{\lambda}{2}D^{3/2}\bar{y}^*_n=(n+1/2)\pi, \quad {\rm and }\quad 0<\bar{y}^*_n<1. \]

  Then, the real and continuous function $\bar{y}\mapsto\rho(\lambda,\bar{y})$ changes sign as $\bar{y}$ varies near $\bar{y}^*_n$. 
\end{proof}

\subsection{Roadmap, Step 7: Eigenvalue splitting distance, $\Delta(\lambda)$, and signed splitting distance $\mathfrak{S}(\lambda)$}\label{sec:roadmap7}

We continue with the double-well potential configuration sketched in \Cref{fig:schematic_2WELL} and given by the expressions \Cref{eq:W2well}, \Cref{eq:V_L}, \Cref{eq:V_R}. The double-well Hamiltonian is given by
\[
H_\lambda = \calP_\lambda^2 + V_\lambda,
\]
where
\begin{align} V_\lambda(x) &= v^{\radial,L} +\tau_{\lambda,M,D}W^L(x) + v^{\radial,R}(x)+\tau_{\lambda,M,D}W^R(x)\\
v^{\radial,R}(x) &= v^{\radial,L}(-x)\quad {\rm and}\quad 
W^R(x) = W^L(-x).
\end{align}

Recall from the introduction that since $[\calI,H_\lambda]=0$,
 where $\calI f(x)=f(-x)$, it is natural to decompose $L^2(\RR^2)$ as in \cref{eq:Hilbert space decomposition into even and odd}. 
We can determine the spectrum of $ H_\lambda$ on $L^2(\RR^2)$ by considering $ H_\lambda$ on these subspaces:
\[
L^2(\RR^2)-{\rm spec}( H_\lambda) = 
 \Big(L^2_{\rm even}-{\rm spec}( H_\lambda)\Big)\
 \cup\ \Big(L^2_{\rm odd}-{\rm spec}( H_\lambda)\Big).
\]

Recall that $E_{0,\lambda}\le E_{1,\lambda}$ denote the lowest eigenvalues of $ H_\lambda$ acting in $L^2(\RR^2)$; see \cref{prop:H2well-low_energy}. Additionally, we have introduced in the introduction $E_{{\rm even},\lambda}$, the ground state of $ H_\lambda$ acting in $L^2_{\rm even}$, and $E_{{\rm odd},\lambda}$, the ground state of $ H_\lambda$ acting in $L^2_{\rm odd}$.

\begin{prop}\label{prop:even-odd-eigs} 
 For $\lambda\ge \lambda_{\min}$ sufficiently large, we have
\[ \{E_{0,\lambda},E_{1,\lambda}\} = \{E_{{\rm even},\lambda}, E_{{\rm odd},\lambda}\}. \]
\end{prop}
\begin{proof}[Sketch of the proof of \cref{prop:even-odd-eigs}] Suppose $\psi\in L^2_{\rm even}$ is orthogonal to $\vf_{\rm even}\equiv \vf^L+\vf^R\in L^2_{\rm even}$.  Then, $\psi$ is orthogonal to both $\vf^L$ and $\vf^R$, since the even function $\psi$ is orthogonal to the odd function $\vf^L-\vf^R$. An energy estimate (see, for example, \cite[Theorem 6.3]{fefferman2025lowerboundsquantumtunneling}) then implies, for $\lambda\ge\lambda_{\rm min}$, that 
\begin{equation}\label{eq:energy-estimate}
    \ip{\psi}{\br{H_\lambda-e_\lambda\Id} \psi} \ge \frac{1}{2}c_{\rm gap}\norm{\psi}^2\, .
\end{equation}
 Hence, the spectrum of $H_\lambda\Big|_{L^2_{\rm even}}$ below $e_\lambda+\frac12 c_{\rm gap}$ consists, at most, of a single eigenvalue. 
On the other hand, one easily verifies that 
\[
\|\vf_{\rm even}\|^2= 2+\mathcal{O}(e^{-c\lambda})\quad {\rm and}\quad 
\ip{\vf_{\rm even}}{\br{H_\lambda-e_\lambda\Id} \vf_{\rm even}} = \mathcal{O}(e^{-c\lambda}),
\]
and  so $H_\lambda\Big|_{L^2_{\rm even}}$ has spectrum below $e_\lambda +  Ce^{-c\lambda} \le e_\lambda+\frac12 c_{\rm gap}$. Therefore, the spectrum of $H_\lambda\Big|_{L^2_{\rm even}}$ below $e_\lambda+\frac12 c_{\rm gap}$ consists of a single eigenvalue $E_{{\rm even},\lambda}$.  A similar argument, with $\vf_{\rm odd}\equiv \vf_L-\vf_R$ in place of $\vf_{\rm even}$, shows that the spectrum of  $H_\lambda\Big|_{L^2_{\rm odd}}$ below 
$e_\lambda+\frac12 c_{\rm gap}$ consists of a single eigenvalue $E_{{\rm odd},\lambda}$. Since $L^2=L^2_{\rm even}\oplus L^2_{\rm odd}$, the 
assertion of \cref{prop:even-odd-eigs} follows.
\end{proof}
%


\begin{rem}[Magnetic vs. Non-magnetic splitting]
 The sign of $\mathfrak{S}(\lambda)$ determines whether the ground state is even (positive sign) or odd (negative sign). 
  In non-magnetic systems, we always have: 
\[ \Delta(\lambda)= E_{{\rm odd},\lambda}-E_{{\rm even},\lambda} \equiv \mathfrak{S}(\lambda)>0 .\]
In this section we display examples in magnetic systems, where $\Delta(\lambda)=0$, and examples where 
 $\Delta(\lambda)>0$ but $\mathfrak{S}(\lambda)<0$.
\end{rem}

The following is a direct consequence of \Cref{prop:even-odd-eigs}: 
\begin{prop}\label{prop:split}
Both $\Delta(\lambda)$ and $\mathfrak{S}(\lambda)$ are determined by the ground state energies, $E_{{\rm even},\lambda}$ and $E_{{\rm odd},\lambda}$, for the spaces $L^2_{\rm even}$ and $L^2_{\rm odd}$.
 In particular, we have
\[ \Delta(\lambda) = \Big|E_{{\rm odd},\lambda} - E_{{\rm even},\lambda}\Big| = \big|\mathfrak{S}(\lambda)\big|. \]
\end{prop}

\subsection{Roadmap, Step 8: Calculation of $E_{{\rm even},\lambda}$ and $E_{{\rm odd},\lambda}$ via LCAO}

  For large $\lambda$, we expect the two dimensional subspace corresponding to the pair of eigenvalues $\{E_{0,\lambda},E_{1,\lambda}\}$
 to be approximately spanned by even $(\sigma=+1)$ and odd $(\sigma=-1) $ linear combinations of $\vf_\lambda^L$ and $\vf_\lambda^R$ (``atomic orbitals''), the ground states of the right and left single-well Hamiltonians, $ h_\lambda^L\vf_\lambda^L = e_\lambda \vf_\lambda^L$ and $ h_\lambda^R\vf_\lambda^R = e_\lambda 
 \vf_R$ (see \Cref{eq:left-well-evp} and \Cref{eq:right-well-evp}):
 \[ \vf_\lambda^L+\vf_\lambda^R \in L^2_{\rm even}= L^2_{+1}\quad\ {\rm and}\quad \vf_\lambda^L-\vf_\lambda^R\in L^2_{\rm odd}=L^2_{-1}. \]
 Note, by \cref{prop:gs-bounds}, $\|\vf_\lambda^L\pm\vf_\lambda^R\|^2\approx 2$.

 For $\sigma=\pm1$, we seek asymptotic expressions for the ground state eigenpairs of $H_\lambda $ acting in the space $L^2_\sigma$. These are non-trivial least energy solutions $(E,\Psi)$ of 
\eql{ H_\lambda \Psi = E \Psi,\quad \Psi\in L^2_\sigma\ ,\ \sigma=\pm1. 
 \label{eq:evp-sigma}}
Motivated by the above discussion, we introduce the $L^2-$ even ($\sigma=+1$) and odd ($\sigma=-1$) states:
\[
\psi_\sigma = 
\frac{\vf_\lambda^L+\sigma \vf_\lambda^R}{\|\vf_\lambda^L+\sigma \vf_\lambda^R\|} \in\ L^2_\sigma
\]
and define the orthogonal projections
\[ \Pi_\sigma f = \left\langle \psi_\sigma, f\right\rangle \psi_\sigma
\quad {\rm and}\quad \Pi_\sigma^\perp = f-\Pi_\sigma f, \]
which project $L^2_\sigma$ to $\CC\psi_\sigma$ and $L^2_{\sigma,\perp}\equiv \Pi^\perp_\sigma L^2_\sigma $, respectively.

Viewing $L^2_\sigma$ as $L^2_{\sigma,\perp}\oplus \CC \psi_\sigma$, we seek solutions of \cref{eq:evp-sigma} of the form:
\eql{\label{eq:sigma-ansatz}
\Psi = \alpha\psi_\sigma + \psi_\sigma^\perp,\quad \Pi^\perp_\sigma\psi_\sigma^\perp=\psi_\sigma^\perp,\quad {\rm and}\quad \alpha\in\CC.
}
Substitution of \cref{eq:sigma-ansatz} into \cref{eq:evp-sigma} and projecting with $\Pi_\sigma$ and $\Pi^\perp_\sigma$, yields the coupled system for
 $(\psi^\perp_\sigma,\alpha)^\top\in L^2_{\sigma,\perp} \oplus \CC$, which is equivalent to \Cref{eq:evp-sigma}:
\eql{\label{eq:psiperp-alpha-sys}
  \begin{pmatrix}\ 
  \Pi^\perp_\sigma \left( H_\lambda -E \right) \Pi^\perp_\sigma & B_\sigma \\
  {\ } & {\ }\\
  B_\sigma^* & \left\langle \psi_\sigma, H_\lambda \psi_\sigma\right\rangle - E 
  \ \end{pmatrix}
  \begin{pmatrix}
    \psi_\sigma^\perp \\ {\ }\\ \alpha
  \end{pmatrix}
  = \begin{pmatrix}
    0 \\ {\ }\\ 0
  \end{pmatrix},
}
where
\begin{subequations}
\label{eq:B-sigma_def}
\begin{align}
 B_\sigma &= \Pi_\sigma^\perp\left( H_\lambda \ \psi_\sigma\right ):\CC\to L^2_{\sigma,\perp} \\
 B_\sigma^* &= \left\langle\psi_\sigma, H_\lambda \ \Pi_\sigma^\perp\ \cdot\right\rangle: L^2_{\sigma,\perp}\to \CC.
\end{align}
\end{subequations}

\noindent The following proposition is a consequence of \cref{eq:energy-estimate} in the proof of \cref{prop:even-odd-eigs}.

\begin{prop}[Spectral gap]\label{prop:spectralgap}
Fix $\varepsilon_0>0$ sufficiently small. 
For any fixed $\lambda\ge\lambda_{\rm min}$ sufficiently large, the following holds for complex $E$ and such that $|E-e_0^\lambda|<\varepsilon_0 c_{\rm gap}$ (see \Cref{assume:gap}):
The operator \[ E\mapsto \Big[ \Pi^\perp_\sigma \left( H_\lambda -E \right) \Pi^\perp_\sigma \Big]^{-1}\] is a well-defined and analytic $\mathcal{B}(\Pi^\perp_\sigma L^2(\RR^2))$-valued function. Moreover, the norm of the operator and its derivatives with respect to $E$ are bounded uniformly with $\lambda$.

\end{prop}

It follows that we can solve the first equation in \Cref{eq:psiperp-alpha-sys} for $\psi^\perp_\sigma$ in terms of $\alpha$:
\eql{
\psi_\sigma^\perp = 
 - \Big[ \Pi^\perp_\sigma \left( H_\lambda -E \right) \Pi^\perp_\sigma \Big]^{-1} B_\sigma\ \alpha
\label{eq:psi-perp-eqn}
}
Substituting \cref{eq:psi-perp-eqn} into the second equation of 
\Cref{eq:psiperp-alpha-sys}, we find 
that for $\lambda\ge\lambda_{\rm miin}$ a non-trivial solution exists if and only if $E$ satisfies the scalar nonlinear equation:

\eq{
E -\left\langle \psi_\sigma, H_\lambda \psi_\sigma\right\rangle
+ B_\sigma^* \Big[ \Pi^\perp_\sigma \left( H_\lambda -E \right) \Pi^\perp_\sigma \Big]^{-1} B_\sigma\ 
 =0.
\nonumber
}
Anticipating a solution $E\approx e_\lambda$, we use that $\psi_\sigma$ is normalized to write:
\eql{ E - e_\lambda -\left\langle \psi_\sigma, \Big( H_\lambda -e_\lambda\Big)\psi_\sigma\right\rangle
+ B_\sigma^* \Big[ \Pi^\perp_\sigma \left( H_\lambda -E \right) \Pi^\perp_\sigma \Big]^{-1} B_\sigma\ 
 =0,
\label{eq:Esigma-eqn}
}
or more explicitly:
\begin{align}
&E\ -\ e_\lambda - \left\langle \psi_\sigma, \big( H_\lambda -e_\lambda\big)\psi_\sigma\right\rangle
\nonumber\\
&\qquad + \Biggl\langle \left(\Pi_\sigma^\perp H_\lambda \psi_\sigma\right), \Big[ \Pi^\perp_\sigma \left( H_\lambda -E \right) \Pi^\perp_\sigma \Big]^{-1} \left(\Pi_\sigma^\perp H_\lambda \psi_\sigma\right)\Biggr\rangle 
 =0\ .
 \label{eq:Esigma-eqn-1}
\end{align}

In the next proposition, we assert that, for $\lambda\ge\lambda_{\rm min}$, the scalar nonlinear equation \cref{eq:Esigma-eqn} for $E$ can be solved uniquely in a neighbhorhood of energy $e_\lambda$ for $E= \ee^{(\sigma)}_{\lambda}$. Then setting $\alpha=1$ in
\cref{eq:psi-perp-eqn} and \cref{eq:sigma-ansatz}
we obtain eigenpairs
$\Psi_\sigma = \psi_\sigma + \psi_\sigma^\perp \in L^2_\sigma$, with corresponding eigenvalues: $E^{(\sigma)} _{\lambda} \approx e_\lambda^{\radial}+\sigma\rho(\lambda)$, where $\sigma=\pm1$ since the nonlinear term in $E$ is negligibly small.

\begin{prop}[Expansions of $E^{(\pm1)} _{\lambda}$ and the splitting, $\Delta(\lambda)$]
\label{prop:Epm1-expand-Delta-expand}
  Let $\lambda\ge\lambda_{\min}$. Then, for $\sigma=\pm1$, \eqref{eq:Esigma-eqn} has a unique solution, $E^{(\sigma)}_{\lambda}$
  in a neighborhood of $e_\lambda$, the ground state energy of $ h_\lambda$. Correspondingly, the double-well eigenvalue problem \cref{eq:evp-sigma} with $E=E^{(\sigma)}_{\lambda}$ has a one-dimensional subspace of solutions spanned by an eigenstate $\Psi_\sigma\in L^2_\sigma$.
 \begin{enumerate} 
 \item Expansions of $E^{(\sigma)}_{\lambda}$, $\sigma=\pm1$ :
 \begin{subequations}
 \label{eq:Esig-expand}
  \begin{align}
  &  \ee^{(\sigma)}_{\lambda} - e_\lambda
  = \sigma\ \rho(\lambda)
     \ +\ 
     {\rm ERR}^{(\sigma)}_\lambda,\\
 &    \Big| {\rm ERR}^{(\sigma)}_\lambda \Big| \lesssim\ C \exp\big(-\frac{\lambda}{2}(D^3-2D)^2+C\lambda D^{3/2} \big)
     \end{align}
     \end{subequations}
     where $\rho(\lambda)$ denotes the magnetic hopping coefficient \cref{eq:hopping_coefficient}.
     
      By  \cref{eq:Esig-expand} and \cref{eq:rho_cos}, we have an
     \item Expansion of the signed magnetic splitting, $\mathfrak{S}(\lambda)\equiv E_{{\rm odd},\lambda} - E_{{\rm even},\lambda} $ as follows:
     \eql{\label{eq:frakS-expanded}
     \mathfrak{S}(\lambda) &= -2\rho(\lambda) \ +\  {\rm ERR}_{\rm split}(\lambda) \nonumber\\
&=2\Bigg(\ 
  \cos\left( \frac{\lambda}{2} D^{3/2}\bar{y} \right) + {\rm ERR}_\rho(\lambda)\ \Bigg)\cdot {\rm FAC}_\rho(\lambda)\cdot 
  \ee^{ -\frac{\lambda}{4}(D^3-D^2) }
   + {\rm ERR}_{\rm split}(\lambda)\nonumber\\
     } where
\eq{
 \big|{\rm ERR}_{\rm split}(\lambda)\big| &= \Big|{\rm ERR}^{(-1)}_\lambda-{\rm ERR}^{(+1)}_\lambda\Big| \lesssim\
  C \exp\big(-\frac{\lambda}{2}(D^3-2D)^2+C\lambda D^{3/2} \big)\,.
}
     \end{enumerate}
\end{prop}
\cref{prop:Epm1-expand-Delta-expand} is proved in \cref{sec:splitting-proof}. 
Given \cref{prop:Epm1-expand-Delta-expand} and \cref{thm:vanishing_hopping}, on the magnetic hopping coefficient $\rho(\lambda)$, we have:
\begin{thm}[Degeneracy of and symmetry class transitions of ground states]
 The vertical displacement $\bar{y}$ can be chosen so that either of the following statements (although not necessarily both simultaneously) hold:
 \begin{itemize}
\item[(a)] The low-energy splitting $\Delta(\lambda)=E_{1,\lambda}-E_{0,\lambda}$ vanishes. In this case, the tunneling time between wells is infinite; tunneling has been eliminated.
\item[(b)] The signed splitting $\mathfrak{S}(\lambda)=E_{{\rm odd},\lambda}-E_{{\rm even},\lambda}$ can be made positive or negative. 
 When $\mathfrak{S}(\lambda)$ is positive, the ground state of $H_\lambda$ is even and when $\mathfrak{S}(\lambda)$ is negative, 
the ground state of $H_\lambda$ is odd.
\end{itemize}
\end{thm}

\noindent {\bf NB:} We note that vanishing or change in sign of the signed splitting, $\mathfrak{S}(\lambda)$, does not necessarily occur at the precise parameter values for which the magnetic hopping coefficient, $\rho(\lambda)$, changes sign.

\section{Implementation of the roadmap of \Cref{sec:results-roadmap}}\label{sec:implement_the_roadmap}
%
%

We begin with a review of the double-well potential, $ V_\lambda$ given schematically in \Cref{fig:schematic_2WELL}, where $\nu_{\rm max}=4$ is the number of ``sophons'' surrounding each centrally located radially symmetric ``planet''.

\noindent {\bf Left potential well configuration:}\\
Left planet centered at: $- d =(-D^{3/2}/2,0)$\\
and surrounding left sophons centered at:
\begin{subequations}
\label{eq:left-centers}
\begin{align*} \zeta_1^L &= \zeta_1- d = (-D^{3/2}/2+D,\bar{y}),\quad 
\zeta_2^L = \zeta_2- d = (-D^{3/2}/2+D,-\bar{y})
\\
\zeta_4^L &=\zeta_4- d = (-D^{3/2}/2-D,\bar{y}),\quad 
\zeta_3^L=\zeta_3- d = (-D^{3/2}/2-D,-\bar{y})
\end{align*}
\end{subequations}

\noindent {\bf Right potential well configuration:}\\
Right planet centering is at: $ d =(D^{3/2}/2,0)$\\
and surrounding right sophons centered at the points:
\begin{align}
\label{eq:right-centers}
 \zeta_j^R = - \zeta_j^L,\quad j=1,2,3,4
 \end{align}

Our goal is to first prove an expansion of the hopping coefficient, $\rho(\lambda)$, (\Cref{prop:hopping_leading+error}) and an expansion of the low-lying eigenvalues of $ H_\lambda$ in terms of $\rho(\lambda)$ (\Cref{prop:even-odd-eigs}).

 \subsection{Proof of \Cref{prop:hopping_leading+error}; Expansion of the hopping coefficient}

 The hopping coefficient, $-\rho$, can be expanded in terms of planet / sophon interaction terms $\Omega(\nu,\nu^\prime)$; see \Cref{eq:rho-expand}, \Cref{eq:Omegas}. \Cref{prop:hopping_leading+error} asserts that for the planet-sophon configuration of 
\Cref{fig:schematic_2WELL}, with appropriately chosen parameters $\delta=\delta_{\lambda,M,D}$ and $\tau=\tau_{\lambda,M,D}$, the dominant contribution to the magnetic hopping, $\rho$, comes from 
 the expressions: $\Omega(1,0)$, $\Omega(0,1)$, $\Omega(2,0)$ and $\Omega(0,2)$, which are displayed in \Cref{eq:Omegas}. 
 
 To prove \Cref{prop:hopping_leading+error} we proceed as follows. The single-well potential $ v_\lambda(x)$ has been taken to be a small perturbation of the radially symmetric single-well potential, $\lambda^2\vzero(x)$.
 Hence the ground state eigenpair $(\vf_\lambda, e_\lambda)$, associated with the potential $ v_\lambda$, is a small perturbation of the ground state eigenpair $(\vf^{\radial}_{\lambda}, e^{\radial}_{\lambda})$, for the radial potential $\lambda^2\vzero$; the estimates are presented in 
\cref{sec:compare-phi_and_phi0}.

 Hence, we first approximate all expressions $\Omega(\nu,\nu^\prime)$ by the expressions $\Omega^\radial(\nu,\nu^\prime)$, which are obtained from $\Omega(\nu,\nu^\prime)$ by replacing perturbed single-well eigenstates ($\vf_\lambda^L$ and $\vf_\lambda^R$) by their unperturbed radial variants 
 ($\vf^{\radial,L}$ and $\vf^{\radial,R}$):
 \begin{align*}
   \vf_\lambda^L(x)&\to \vf^{\radial,L}(x)=e^{i\frac{\lambda}{2}(x\wedge d )}\vf^{\radial}_{\lambda}(x+ d ),\\
   \vf_\lambda^R(x)&\to \vf^{\radial,R}(x)=e^{-i\frac{\lambda}{2}(x\wedge d )}\vf^{\radial}_{\lambda}(-x+ d ) = \ee^{-i\frac{\lambda}{2}(x\wedge d )}\vf^{\radial}_{\lambda}(x- d ),
 \end{align*}
 where the last equality holds by the assumption that $\vf^{\radial}_{\lambda}$ is radial.

Thus, in parallel with the expression for $-\rho(\lambda)$ in \Cref{eq:rho-expand}, we introduce an\\
\underline{\it approximate magnetic hopping}:
 \eql{\label{eq:rho-0-expand}
-\rho^{\radial}(\lambda)\ =\ \Omega^\radial(0,0) + \sum_{\nu=1}^{\nu_{\rm max}}\Big(\ \Omega^\radial(\nu,0)\ +\ \Omega^\radial(0,\nu)\ \Big)\ +\ \sum_{\nu,\nu^\prime=1}^{\nu_{\rm max}} \Omega^\radial(\nu,\nu^\prime),
}
and we therefore write
\[
-\rho(\lambda)\ =\ -\rho^{\radial}(\lambda)\ + \left( \rho^{\radial}(\lambda) - \rho(\lambda)\right).
\]
 Based on expansions of $\Omega^\radial(0,\nu)$, and also \eq{\Omega^\radial(\nu,0)=\overline{\Omega^\radial(0,\nu)}\qquad(\nu=1,2)\,,} and upper bounds on all other 
 $\Omega^\radial(\nu,\nu^\prime)$ terms, we shall choose $\tau_{\lambda,M,D}$ and $\delta_{\lambda,M,D}$ (see \cref{eq:delta&tau}) to satisfy upper and lower bound constraints consistent with our ordering hypotheses \Cref{eq:PS-PP} and \Cref{eq:PS-SS}, uniformly in $0\le \bar{y}\le 1$, for any fixed $D\ge D_{\rm max}$ and $\lambda>\lambda_{\rm max}$ sufficiently large. This specification of parameters: $D$, $\delta_{\lambda,M,D}$, $\tau_{\lambda,M,D}$ and $\bar{y}$ completely determines the double-well configuration.

We then prove in \Cref{eq:rho0-expression} that the large $\lambda$ asymptotic behavior of $-\rho^{\radial}(\lambda)$ is given by the expression in \Cref{eq:rho_cos}. Finally, in \Cref{sec:from_rho0_2_rho} we 
transition from radial atomic ground states to perturbed atomic ground states and show that the contribution to $-\rho(\lambda)$ from $\rho^{\radial}(\lambda) - \rho(\lambda)$ is negligible.
This will complete the proof of \Cref{prop:hopping_leading+error}.

 \subsubsection{The approximation $\Omega^\radial(0,\nu)=\Omega^\radial( P_R \leftrightarrow S_{L,\nu})$}\label{sec:Omega0nu0-expand}
 An important step in expanding $\rho$ in terms of dominant and subdominant contributions is the following : 
 \begin{prop}\label{prop:sophon0-planet0}
  Fix $\nu\in\{1,2,3,4\}$. Let $\Omega^\radial(0,\nu)=\Omega^\radial( P_R \leftrightarrow S_{L,\nu})$ denote the following expression, $\Omega^\radial(0,\nu)$, obtained from $\Omega(0,\nu)$ in \cref{sec:RM-Step5} by replacing the exact single-well states $\vf_\lambda^L$ and $\vf_\lambda^R$ by the approximations $\vf^{\radial,L}$ and $\vf^{\radial,R}$, respectively :
 \eql{  \Omega^\radial(0,\nu) \equiv \int_{\RR^2\times\RR^2} 
 \lambda^2 v^{\radial,R}(x) \overline{\vf^{\radial,R}(x)}\ \ee^{-i\frac{\lambda}2 x\wedge y}\ K(x-y)\ \tau_{\lambda,M,D}W^L_\nu(y) \vf^{\radial,L}(y)\ \dif{x} \dif{y}\ .
 \label{eq:Omega00nu}
 }
   Then,
  \begin{align} 
  \Omega^\radial(0,\nu)
= \left( 1 + {\rm ERR}_\nu \right) \ \Omega^\radial_{\rm main}(0,\nu), 
\label{eq:Om0-main-1} \end{align}
where
\eql{
 \Omega^\radial_{\rm main}(0,\nu) =
 -\ee^{-\ii\lambda\ d \wedge\zeta_\nu}\ \int_{|y|\le\delta}\tau_{\lambda,M,D}|W^0(y)| \vf^{\radial}(y+\zeta_\nu)\ 
  \tilde{\rho}_{0,\lambda}(2 d -\zeta_\nu-y)\ \dif{y},
\label{eq:Omega0main}
}
\eql{
\exp\left(-\frac{\lambda}{4}|2 d -\zeta_\nu|^2-C\lambda| d |\right)\ \le\ 
 -\tilde{\rho}_{0,\lambda}(2 d -\zeta_\nu-y)\ \le\
 \exp\left(-\frac{\lambda}{4}|2 d -\zeta_\nu|^2+C\lambda| d |\right), 
\label{eq:trho-upper-lower-1}}
and 
\[ \big|{\rm ERR}_\nu\big| \le \exp\left(-\frac{M}{3}\lambda D^{3/2}\right).\]\nc
 \end{prop}

\noindent We shall further expand the expression for $\Omega^\radial_{\rm main}(0,\nu)$, \Cref{eq:Omega0main}, in \Cref{prop:dominant-Omega0nu} below and thereby obtain a detailed expansion of $\Omega^\radial(0,\nu)$.\bigskip

 In preparation for the proof of \Cref{prop:sophon0-planet0}, we study the following key overlap integral, closely related to the hopping coefficient for double-wells with radially symmetric atomic potentials. 
%
 \begin{align}
   \tilde{\rho}_{0,\lambda}(Y)
   &\equiv \int \ee^{-i\frac{\lambda}{2}\tilde{x}\wedge Y} K(Y-\tilde{x})\ (\lambda^2v^{\radial})(\tilde{x})\ \vf^{\radial}(\tilde{x})\ d\tilde{x}\nonumber \\
  &= \int \ee^{-i\frac{\lambda}{2}x\wedge Y} K\Big(x-\frac{Y}{2}\Big)\ (\lambda^2v^{\radial})\Big(x+\frac{Y}{2}\Big)\ \vf^{\radial}\Big(x+\frac{Y}{2}\Big)\ \dif{x}
   \label{eq:tilde-rho}
 \end{align}

 By \Cref{eq:phi0-K}, $\vf^{\radial}(y)=\Gamma K(y)$ for $|y|>1$, where $\Gamma=\Gamma_\lambda$ denotes a positive constant satisfying $e^{-C\lambda}\le \Gamma\le \ee^{C\lambda}$; see \Cref{eq:Gamma-bounds}. If $|Y|>4$ and $x\in {\rm supp}(\vzero)\subset B_1(0)$ we have
 $|x-Y/2|>1$, and therefore
  $\vf^{\radial}(x-Y/2)=\Gamma K(x-Y/2)$. 
  Next, multiplying \Cref{eq:tilde-rho} by $\Gamma$ yields: 
 \eql{\label{eq:rho-trho} \rho_{0,\lambda}(Y) = \Gamma\ \tilde{\rho}_{0,\lambda}(Y) ,}
where $\rho_{0,\lambda}(Y)$ denotes the real-valued magnetic hopping coefficient for the double-well potential, comprised of radial atomic wells $\lambda^2 v^{\radial}(x+Y/2)$ and $\lambda^2 v^{\radial}(x-Y/2)$;
 see \Cref{prop:rho-real}. Hence, $\tilde{\rho}_{0,\lambda}(Y)=\Gamma^{-1}\rho_{0,\lambda}(Y) $ is real-valued.

\begin{lem}[Main overlap integral lemma]
\label{lem:main-overint}
The function $Y\mapsto \rho_{0,\lambda}(Y)$ is real-valued and negative. Moreover, for a fixed $C$ sufficiently large we have that if $|Y|>C$, then 
 \eql{\label{eq:rho-tilde-bound}
 \exp\left(-\frac{\lambda}{4}|Y|^2 -C\lambda|Y| \right)\ \le\ -\tilde{\rho}_{0,\lambda}(Y)\ \le\ 
 \exp\left(-\frac{\lambda}{4}|Y|^2+C\lambda|Y| \right)
 }
\end{lem}
\begin{proof}[Proof of \Cref{lem:main-overint} ]
The proof is based on results and calculations from \cite{FSW_22_doi:10.1137/21M1429412}.
By \Cref{eq:rho-trho}, the bounds \Cref{eq:rho-tilde-bound} on $\tilde{\rho}_{0,\lambda}(Y)$ is equivalent to the bounds
\eql{
 \exp\left(-\frac{\lambda}{4}|Y|^2 -C\lambda|Y| \right)\ \le\ -\rho _{0,\lambda}(Y)\ \le\ 
 \exp\left(-\frac{\lambda}{4}|Y|^2+C\lambda|Y| \right),
 \label{eq:radial-hopping-bounds} 
 }
 for $|Y|>C$, where the constant $\Gamma$ can be absorbed in the bound by taking $C$ sufficiently large. Recall that $\vzero$ and $\vf^{\radial}$ are rotationally invariant. Since for any rotation matrix $R$, we have $x\wedge RY= R^*x\wedge Y$, it follows that 
 $\rho_{0,\lambda}(Y)=\rho_{0,\lambda}(d)$, where $d=(|Y|,0)$. 
 Applying \cite[Eqn (5.1)-(5.2)]{FSW_22_doi:10.1137/21M1429412}, we have that 
 \eq{
   \rho _{0,\lambda}(Y) = \int_0^1 \vf^{\radial}(r) \lambda^2 \vzero(r) L_{|Y|}(r) dr,\quad \textrm{where $L_{|Y|}(r)>0$}.
 }
 Further, since $\vzero\le0$ and $\vf^{\radial}>0$, it follows that $\rho _{0,\lambda}(Y)<0$. 
 Moreover Theorem 1.7 of \cite{FSW_22_doi:10.1137/21M1429412} implies 
$ -\rho _{0,\lambda}(Y)\ge \exp\left( -\frac{\lambda}{4}|Y|^2-C\lambda|Y|\right)$.
This establishes the more delicate lower bound in \Cref{eq:radial-hopping-bounds}.

To complete the proof of \Cref{lem:main-overint}, we shall prove the upper bound \eql{\label{eq:rho0-upper}
|\rho _{0,\lambda}(Y)|\le \exp\left( -\frac{\lambda}{4}|Y|^2+C\lambda|Y|\right),} for $|Y|>C$ and $C$ sufficiently large. 

Since $\vf^{\radial}_{\lambda}(x)=\Gamma K(x)$ (see \Cref{eq:phi0-K}), we have 
\[\rho_{0,\lambda}(Y)
   \equiv \int \ee^{-i\frac{\lambda}{2}\tilde{x}\wedge Y}\ \Gamma\ K(Y-\tilde{x})\ (\lambda^2v^{\radial})(\tilde{x})\ \vf^{\radial}(\tilde{x})\ d\tilde{x}\]
 Since $0\le K(x)\le C\exp\left(-\frac{\lambda}{4}\norm{x}^2\right)$ (see \cref{lem:bounds on the SHO Green's function}) 
 for all $x\in B_C(0)^c$ and  $e^{-C\lambda}\le \Gamma\le \ee^{C\lambda}$ (eqn. \Cref{eq:Gamma-bounds}), we have $0\le 
 \Gamma\ K(Y-\tilde{x})\le C\ \exp\left(-\frac{\lambda}{4}|Y|^2 + C\lambda |Y|\right)$.  
 Therefore, 
 \eq{ |\rho_{0,\lambda}(Y)|\le \int_{\mathbb R^2} \lambda^2|\vzero(x)| \ \vf^{\radial}_{\lambda}(x)\ \dif{x}\ 
 \exp\left(-\frac{\lambda}{4}|Y|^2 + C\lambda |Y|\right).
 \nonumber}
 Finally, ${\rm supp}(\vzero)\subset B_1(0)$, $|\vzero(x)|\le C_0$ and $\|\vf^{\radial}_{\lambda}\|_{L^2(\mathbb R^2)}=1$ imply the bound \cref{eq:rho0-upper}.
 \end{proof}

\begin{proof}[Proof of \Cref{prop:sophon0-planet0}] 
 We now embark on the proof of \Cref{prop:sophon0-planet0} concerning the expansion of $\Omega^\radial(0,\nu)$. 
 We start by expressing $\Omega^\radial(0,\nu)$ in terms of the untranslated potentials and ground states. Note:
 \begin{align*}
   v^{\radial,R}(x)&=v^{\radial}(x- d );\ \vf^{\radial,R}(x)=e^{-i\frac{\lambda}{2}x\wedge d }\vf^\radial(x- d );\\ W_\nu^L(y)&=W^0(y+ d -\zeta_\nu);\ \vf^{\radial,L}(y)= 
   \ee^{i\frac{\lambda}{2}y\wedge d }\vf^\radial(y+ d ).
    \end{align*}
    Hence,
    {\footnotesize{
    \begin{align*} & \Omega^\radial(0,\nu) \\
   &\quad = \int_{\RR^2\times\RR^2}\ \ee^{-i\frac{\lambda}{2}\left( x\wedge y -x\wedge d - y\wedge d \right)}\ 
 \lambda^2 v^{\radial}(x- d ) \vf^{\radial}(x- d )\ K(x-y)\ \tau_{\lambda,M,D}W^0(y+ d -\zeta_\nu)\ \vf^{\radial}(y+ d )\ \dif{x} \dif{y}
 \end{align*}
 }}
 Change variables:
 \begin{align*}
  \tilde{x}&=x- d ;\ \tilde{y}=y+ d -\zeta_\nu,\quad {\rm or}\\
  x&=\tilde{x}+ d ;\ y=\tilde{y}- d +\zeta_\nu,
 \end{align*}
 and note
 \begin{align*}
&   x-y=\tilde{x}-\tilde{y}+2 d -\zeta_\nu;\\
&x\wedge y-x\wedge d -y\wedge d =
2 d \wedge\zeta_\nu+\tilde{x}\wedge(\zeta_\nu-2 d )+(\tilde{x}+2 d )\wedge\tilde{y}.  
 \end{align*}
 Hence, 
 \begin{align}
 &\Omega^\radial(0,\nu)\ = \ \Omega^\radial_{\rm main}(0,\nu)\ +\ \Omega^\radial_{\rm error}(0,\nu),\label{eq:Omega0-main+error}
 \end{align}
 where 
 {\footnotesize{
 \begin{align}
 &e^{i\lambda\ d \wedge\zeta_\nu}\ \Omega^\radial_{\rm main}(0,\nu) \nonumber\\ 
 &\equiv \int_{\mathbb R^2} \tau_{\lambda,M,D}\ W^0(\tilde{y})\ \vf^{\radial}(\tilde{y}+\zeta_\nu) \nonumber\\
 &\qquad \cdot \Big[ \int_{\mathbb R^2}\ \ee^{+i\frac{\lambda}{2}\left(\tilde{x}\wedge(2 d -\tilde{y}-\zeta_\nu)\right)}\ K\left([2 d -\tilde{y}-\zeta_\nu]+\tilde{x}\right) \ \lambda^2 v^{\radial}(\tilde{x}) \vf^{\radial}(\tilde{x})\ d\tilde{x}\Big]\ d\tilde{y} \label{eq:Omega0main1}\\
 &{\rm and} \nonumber\\
 & \ee^{i\lambda\ d \wedge\zeta_\nu}\ \Omega^\radial_{\rm error}(0,\nu) \nonumber\\ 
&\equiv \int_{\mathbb R^2} \tau_{\lambda,M,D}\ W^0(\tilde{y})\ \vf^{\radial}(\tilde{y}+\zeta_\nu) \nonumber\\
&\ \ \cdot \Big[ \int_{\mathbb R^2}\ \ee^{i\frac{\lambda}{2}\left(\tilde{x}\wedge(2 d -\tilde{y}-\zeta_\nu)\right)}\ K\left([2 d -\tilde{y}-\zeta_\nu]+\tilde{x}\right) \ 
 \left( \ee^{-i\lambda d \wedge\tilde{y}}-1\right)\ \lambda^2 v^{\radial}(\tilde{x}) \vf^{\radial}(\tilde{x})\ d\tilde{x}\Big]\ d\tilde{y} \label{eq:Omega0error}
 \end{align}
 }}

 \subsection*{\underline{The term $e^{i\lambda\ d \wedge\zeta_\nu}\ \Omega^\radial_{\rm main}(0,\nu)$ } }

By \Cref{eq:tilde-rho} we have
\begin{equation}\label{eq:Omega0main2}
\begin{aligned}
e^{i\lambda\, d\wedge \zeta_\nu}\,\Omega^{0}_{\text{main}}(0,\nu)
&= \int_{\mathbb{R}^2} \tau_{\lambda,M,D}|W^0(y)| \,\varphi^{0}(y+\zeta_\nu) \\
&\qquad \times (-\tilde{\rho}_{0,\lambda})(2d-\zeta_\nu-y)\, \dif{y}>0\,,
\end{aligned}
\end{equation}
 where we have used the assumption $W^0\le0$. 
 
 Each factor in the integrand of \cref{eq:Omega0main2} is non-negative, so we can bound this integral from above and below by applying  \Cref{lem:main-overint} to obtain upper and lower bounds for
 $-\tilde{\rho}_{0,\lambda}(Y)$,
with $Y=2 d -\zeta_\nu-y$. Note $|\zeta_\nu|\le CD$ and for $y\in \supp(W^0)$, $|y|\le \delta_{\lambda,M,D}$. 
 And since $|2 d |=D^{3/2}$, we have

$|2 d -\zeta_\nu-y|\ge c(D^{3/2}-CD)$ 

so for $D$ sufficiently large, 
$|2 d -\zeta_\nu-y|>C$. Therefore,

\eql{
\exp\left(-\frac{\lambda}{4}|2 d -\zeta_\nu|^2-C\lambda| d |\right)\ \le\ 
 \left(-\tilde{\rho}_{0,\lambda}\right)(2 d -\zeta_\nu-y)\ \le\
 \exp\left(-\frac{\lambda}{4}|2 d -\zeta_\nu|^2+C\lambda| d |\right).
\label{eq:trho-upper-lower}}
This implies upper and lower bounds for the positive quantity:
 $e^{i\lambda\ d \wedge\zeta_\nu}\ \Omega^\radial_{\rm main}(0,\nu)$.

 \subsection*{\underline{The term $e^{i\lambda\ d \wedge\zeta_\nu}\ \Omega^\radial_{\rm error}(0,\nu)$, given by \Cref{eq:Omega0error}} }

Recall from \Cref{eq:delta&tau} and \Cref{eq:xiD} that 
\eql{
  \delta_{\lambda,M,D} =\exp(-M\lambda D^{3/2})\quad {\rm and}\quad 2| d |=D^{3/2},\label{eq:delta-1}
}
where $M$ will presently be determined.
%
%
 For $\tilde{y}\in {\rm supp}(W^0)\subset B_{\delta}(0)$ (see (S2) ), we have the bound:
\[
|e^{-i\lambda d \wedge\tilde{y}}-1|
\le C \lambda | d |\ |\tilde y|\le 
C \lambda D^{3/2} \exp(-M\lambda D^{3/2})\le \exp\left(-\frac{M}{2}\lambda D^{3/2}\right).
\]

\begin{rem}[Choice of $\delta_{\lambda,M,D}$]\label{rem:delta-rationale}
  It is here that we used the choice $\delta_{\lambda,M,D}=\exp(-M\lambda D^{3/2})$. This ensures that the product $\big|\Omega^\radial_{\rm error}\big| \cdot \big|\Omega^\radial_{\rm main}\big|^{-1}$ is exponentially small.
\end{rem}

For $\tilde{x}\in {\rm supp}(\vzero)\subset B_1(0)$ and 
$\tilde{y}\in {\rm supp}(W^0)\subset B_\delta(0)$
we have that $|2 d -\zeta_\nu-\tilde{y}- \tilde{x}|>C$. Therefore, by \cref{lem:bounds on the SHO Green's function} 
\[
K\left(2 d -\zeta_\nu-\tilde{y}-\tilde{x}\right)\le \exp\left(-\frac{\lambda}{4}|2 d -\zeta_\nu-\tilde{y}|^2+C\lambda| d |\right). 
\]
Further, $ \|\lambda^2\ \vzero\ \vf^{\radial}_{\lambda}\|_{L^2(\mathbb R^2)} \le C \lambda^2$, 
since ${\rm supp}(\vzero)\subset B_1(0)$, $|\vzero|\le C$ and $\|\vf^{\radial}_{\lambda}\|_{L^2(\mathbb R^2)}=1$.
Consequently, the integral within the square brackets of \cref{eq:Omega0error} is dominated by 
\begin{align}
 &  C\lambda^2\cdot \exp\left(-\frac{M}{2}\lambda D^{3/2}\right)\cdot \exp\left(-\frac{\lambda}{4}|2 d -\zeta_\nu-\tilde{y}|^2+C\lambda| d |\right) \nonumber\\
 &\qquad \le \exp\left(-\frac{M}{3}\lambda D^{3/2}\right)\cdot \exp\left(-\frac{\lambda}{4}|2 d -\zeta_\nu-\tilde{y}|^2-C\lambda| d |\right),
 \label{eq:exp-bound-in-Om_error}
\end{align}
where by taking $M$ fixed and sufficiently large we can ensure that the constant $C$ on the right hand side of \cref{eq:exp-bound-in-Om_error} is as in
\cref{eq:trho-upper-lower}. 
It follows that the integral within the square brackets of \cref{eq:Omega0error} is bounded by $\exp\left(-\frac{M}{3}\lambda D^{3/2}\right)\cdot \left(-\tilde{\rho}_{0,\lambda}\right)(2 d -\zeta_\nu-y)$.
 Hence, 
\begin{align} \Big|\Omega^\radial_{\rm error}(0,\nu)\Big|
&\le \exp\left(-\frac{M}{3}\lambda D^{3/2}\right)\cdot \int_{\mathbb R^2} \tau_{\lambda,M,D}\ |W^0(\tilde{y})|\ \vf^{\radial}(\tilde{y}+\zeta_\nu)\ 
\left(-\tilde{\rho}_{0,\lambda}\right)(2 d -\zeta_\nu-y)\ d\tilde{y} \nonumber\\
&\le 
\exp\left(-\frac{M}{3}\lambda D^{3/2}\right)\cdot \Big|\Omega^\radial_{\rm main}(0,\nu)\Big| = \exp\left(-\frac{M}{3}\lambda D^{3/2}\right)\cdot \ee^{i\lambda\ d \wedge\zeta_\nu}\ \Omega_{\rm main}^0(0,\nu),\nonumber\\
&{\ }\label{eq:error-main}
\end{align}
Using this bound and 
\begin{align*}
e^{i d \wedge\zeta_\nu}\ \Omega^\radial(0,\nu) &= \ee^{i d \wedge\zeta_\nu}\ \Omega_{\rm main}^0(0,\nu)\ +\ \ee^{i d \wedge\zeta_\nu}\ \ee^{i\varphi_{\rm error}}\ |\Omega_{\rm error }^0(0,\nu)|
\end{align*}
we have 
\begin{align*}  \Big|\ \ee^{i d \wedge\zeta_\nu}\ &\Omega^\radial(0,\nu)\ -\ \ee^{i d \wedge\zeta_\nu}\ \Omega_{\rm main}^0(0,\nu)\ \Big|\\
&\le\ |\Omega_{\rm error }^0(0,\nu)|\ \le\ 
e^{-\frac{M}{3}\lambda D^{3/2}}\cdot \ee^{i\lambda\ d \wedge\zeta_\nu}\ \Omega_{\rm main}^0(0,\nu) 
\end{align*}
Therefore, 
\eql{ \ee^{i d \wedge\zeta_\nu}\ \Omega^\radial(0,\nu)
= \ee^{i d \wedge\zeta_\nu}\ \Omega^\radial_{\rm main}(0,\nu)\ \left( 1 + {\rm ERR}_\nu \right)\ .
\label{eq:Omega0main+error}
}
Here, (see \Cref{eq:Omega0main2})
\begin{subequations}\label{eq:Omega0main22}
\begin{align}
e^{i\lambda\, d\wedge\zeta_\nu}\,\Omega^\radial_{\rm main}(0,\nu)
&= \int_{|y|\le \delta} \tau_{\lambda,M,D}|W^0(y)|\, \vf^{\radial}(y+\zeta_\nu) \notag\\
&\quad\times \big(-\tilde{\rho}_{0,\lambda}\big)(2d-\zeta_\nu-y)\,\dif{y},
\label{eq:Omega0main2-a}\\
\intertext{and}
\big|{\rm ERR}_\nu\big|
&\le \exp\!\left(-\frac{M}{3}\lambda D^{3/2}\right).
\label{eq:ERR_nu-est}
\end{align}
\end{subequations}
 The proof of \Cref{prop:sophon0-planet0} is now complete.\end{proof}
 
Let's now delve further into \Cref{eq:Omega0main+error} by next deriving upper and lower bounds on the expression for 
 $e^{i\lambda\ d \wedge\zeta_\nu}\ \Omega^\radial_{\rm main}(0,\nu)$ in \Cref{eq:Omega0main2-a}.
 
 We first give upper and lower bounds for 
 $\vf^{\radial}(\tilde{y}+\zeta_\nu)$. 
 Recall that for $\vf^{\radial}(x)=\Gamma K(x)$
 for $\norm{x}>1$ (\Cref{eq:phi0-K} ), where $e^{-C\lambda} \le\ \Gamma \le \ee^{C\lambda}$; see \Cref{eq:Gamma-bounds}.
 Moreover, by \Cref{lem:bounds on the SHO Green's function}, for $|\tilde{y}|<\delta$,
\[\exp\left(-\frac{\lambda}{4}|\zeta_\nu|^2 - C\lambda|\zeta_\nu|\right)
  \le K(\tilde{y}+\zeta_\nu) \le \exp\left(-\frac{\lambda}{4}|\zeta_\nu|^2 + C\lambda|\zeta_\nu|\right),\quad \norm{x}>C_1\ .\] 
  
  Hence,
  \[
  \exp\left(-\frac{\lambda}{4}|\zeta_\nu|^2 - C\lambda|\zeta_\nu| \right)
  \le \vf^{\radial}_{\lambda}(\tilde{y}+\zeta_\nu) \le \exp\left(-\frac{\lambda}{4}|\zeta_\nu|^2 + C\lambda|\zeta_\nu| \right) \ .
  \]
  Therefore since, for $\nu=1,2,3,4$, $\Big| |\zeta_\nu|-D\Big|\le C$, we have
  \eql{
  \exp\left(-\frac{\lambda}{4}D^2 - C\lambda D \right)
  \le \vf^{\radial}_{\lambda}(\tilde{y}+\zeta_\nu) \le \exp\left(-\frac{\lambda}{4}D^2 + C\lambda D \right) \ .
  \label{eq:phi0Ypluszeta}
  }

  Applying \Cref{eq:phi0Ypluszeta} together with the upper and lower bounds \Cref{eq:trho-upper-lower} for $\left(-\tilde{\rho}_{0,\lambda}\right)(2 d -\zeta_\nu-y)$ yields upper and lower bounds $e^{i\lambda\ d \wedge\zeta_\nu}\ \Omega^\radial_{\rm main}(0,\nu)$:\\
For $\nu=1,2$, 
  { \footnotesize{
  \begin{align} 
  &\ \tau\delta^2\ \fint_{B_\delta(0)}|W^0|\
  \ee^{-\frac{\lambda}{4}D^2-C\lambda D -\frac{\lambda}{4}(D^{3/2}-D)^2-C\lambda D^{3/2}}\ \nonumber \\ &\qquad \le  
 \ee^{i\lambda\ d \wedge\zeta_\nu}\ \Omega^\radial_{\rm main}(0,\nu) \le \tau\delta^2\ \fint_{B_\delta(0)}|W^0|\
  \ee^{-\frac{\lambda}{4}D^2+C\lambda D -\frac{\lambda}{4}(D^{3/2}-D)^2+C\lambda D^{3/2}},\label{eq:Omega0nu12}\end{align}}}
  and for $\nu=3,4$: 
  { \footnotesize{
  \begin{align} 
  &\ \tau\delta^2\ \fint_{B_\delta(0)}|W^0|\
  \ee^{-\frac{\lambda}{4}D^2-C\lambda D -\frac{\lambda}{4}(D^{3/2}+D)^2-C\lambda D^{3/2}}\ \nonumber \\ &\qquad \le\ \ee^{i\lambda\ d \wedge\zeta_\nu}\ \Omega^\radial_{\rm main}(0,\nu)\ \le\ \tau\delta^2\ \fint_{B_\delta(0)}|W^0|\
  \ee^{-\frac{\lambda}{4}D^2+C\lambda D -\frac{\lambda}{4}(D^{3/2}+D)^2+C\lambda D^{3/2}}\ .\label{eq:Omega0nu34}\end{align}}}
Here, we have used the assumptions of \Cref{sec:single-well-plus-sophons}:
 $W^0\le0$ (Assumption S1) and 
 $\int_{B_\delta(0)}W^0=-c\delta^2$ 
 (Assumption S3), and that $0\le \bar{y}\le1$. 

 Introduce $R_\nu(D,\lambda)$ via 
 \begin{equation}\label{eq:RDlam-nu12}
\begin{aligned}
e^{i\lambda\, d\wedge \zeta_\nu}\,\Omega^\radial_{\rm main}(0,\nu)
&= \tau_{\lambda,M,D}\delta^2 \fint_{B_\delta(0)} |W^0| \\
&\quad\times
e^{-\frac{\lambda}{4}D^2
   - \frac{\lambda}{4}\bigl(D^{3/2}-D\bigr)^2}\,
R_\nu(D,\lambda),\\
&\text{for }\nu=1,2.
\end{aligned}
\end{equation}
 and 
 \begin{equation}\label{eq:RDlam-nu34}
\begin{aligned}
e^{i\lambda\, d\wedge \zeta_\nu}\,\Omega^\radial_{\rm main}(0,\nu)
&= \tau_{\lambda,M,D}\delta^2 \fint_{B_\delta(0)} |W^0| \\
&\quad\times
e^{-\frac{\lambda}{4}D^2
   - \frac{\lambda}{4}\bigl(D^{3/2}+D\bigr)^2}\,
R_\nu(D,\lambda),\\
&\text{for }\nu=3,4.
\end{aligned}
\end{equation}
 Substituting \Cref{eq:RDlam-nu12} and \Cref{eq:RDlam-nu34} 
 into \Cref{eq:Omega0nu12} and  \Cref{eq:Omega0nu34}, respectively, we find $e^{-C\lambda D^{3/2}}\ \le\ 
   R_\nu(D,\lambda)\ \le \ee^{ C\lambda D^{3/2}}$.
Summarizing, we have the following expansions for $\Omega^\radial(0,\nu)$, $\nu=1,2,3,4$:
 \begin{prop}[Expansion of $\Omega^\radial(0,\nu)$]\label{prop:dominant-Omega0nu}
\begin{enumerate}
\item For $\nu=1,2$
 \begin{align}
  \ee^{i\lambda\ d \wedge\zeta_\nu}\ \Omega^\radial(0,\nu) = \tau\delta^2 \times 
\fint_{B_\delta(0)}|W^0|\ 
 \ee^{-\frac{\lambda}{4}D^2 -\frac{\lambda}{4}(D^{3/2}-D)^2}\times R_\nu(D,\lambda)\times \left[ 1 + {\rm ERR}_\nu\right], \nonumber \\
 &{\ }\label{eq:Omega00nu-exp}
 \end{align}
 \item 
 For $\nu=3,4$
 \begin{align}
 & \ee^{i\lambda\ d \wedge\zeta_\nu}\ \Omega^\radial(0,\nu) = \tau\delta^2 \times 
\fint_{B_\delta(0)}|W^0|\ 
 \ee^{-\frac{\lambda}{4}D^2 -\frac{\lambda}{4}(D^{3/2}+D)^2}\times R_\nu(D,\lambda)\times \left[ 1 + {\rm ERR}_\nu\right]\ .
\nonumber \\
 &  {\ }
 \label{eq:Omega00nu-exp-12}
 \end{align}
 \end{enumerate}
Here, for $\nu=1,2,3,4$, we have 
 \begin{align}
   &\exp\left( -C\lambda D^{3/2}\right)\ \le\ 
   R_\nu(D,\lambda)\ \le \exp\left( C\lambda D^{3/2}\right),
   \label{eq:RDlam-bounds-34}
   \end{align}
   and 
   \begin{align}
& \big|{\rm ERR}_\nu\big| \le \exp\left(-\frac{M}{3}\lambda D^{3/2}\right)\ . \label{eq:ERR_nu-bd}
 \end{align}
   \end{prop}

 \subsection*{Upper bound on $\Omega^\radial(0,0)=
 \Omega^\radial(P^R\leftrightarrow P^L)$}

We have
\[
\Omega^\radial(0,0) = \int_{\RR^2\times\RR^2} 
 \lambda^2 v^{\radial,R}(x) 
 \overline{\vf_\lambda^{\radial,R}(x)}\ \ee^{-i\frac{\lambda}2 x\wedge y}\ K(x-y)\  \lambda^2 v^{\radial,L}(y) \vf_\lambda^{\radial,L}(y)\ \dif{x} \dif{y}.
 \]
In the integrand of $\Omega^\radial(0,0)$, $x$ varies over the support of $v^{\radial,R}$ and $y$ over the support of $v^{\radial,L}$, i.e. $|x- d |\le1$ and $|y- d |\le1$. Hence, $|x-y|\ge 2| d |-2=D^{3/2}-2$. Therefore, by \Cref{lem:bounds on the SHO Green's function}, if we take $D>C_1$ we have
 $K(x-y)\le C \ee^{-\frac{\lambda}{4}|x-y|^2 }\le C 
 \ee^{-\frac{\lambda}{4}D^3 + C\lambda D^{3/2}}$
 Using this bound on $K(x-y)$, the Cauchy-Schwarz inequality, and the $L^2-$ normalization of $\vf^{\radial}$, we have 
 \begin{align}
   |\Omega^\radial(0,0)| &\le \lambda^4\ \|\vzero\|_\infty^2\ \int_{\substack{|x- d |\le1\\ |y+ d |\le1}}\ \vf^{\radial}(x- d )\ 
 K(x-y)\ \vf^{\radial}(y+ d )\ \dif{x} \dif{y} \le C\lambda^4 \ee^{-\frac{1}{4}D^3 + C\lambda D^{3/2}}\nonumber 
 \end{align}

 We have proved
 \begin{prop}\label{prop:Omega00-bd}
 \eql{
   |\Omega^\radial(0,0)|\le C\lambda^4 \ee^{-\frac{\lambda}{4}D^3 + C\lambda D^{3/2}}. \label{eq:Omeg000-bd-1}
 }
 \end{prop}

\subsection*{Upper bounds on $\Omega^\radial(\nu,\nu^\prime)=\Omega^\radial(S^R_{\nu^\prime}\leftrightarrow S^L_\nu)$ for $\nu,\nu^\prime=1,2,3,4$ }

We have 
  \[ \Omega^\radial(\nu^\prime,\nu) = \tau^2 \ \int_{\RR^2\times\RR^2} 
 W^R_{\nu^\prime}(x) \overline{\vf^{\radial,R}(x)}\ \ee^{-i\frac{\lambda}2 x\wedge y}\ K(x-y)\ W^L_\nu(y) \vf^{\radial,L}(y)\ \dif{x} \dif{y}
 \]
\noindent 
Recall from \Cref{eq:WRLnu} that 
 $W^L_{\nu}(y)=W^0\big(y-(\zeta_\nu- d )\big)$ and $W^R_{\nu^\prime}(x)=W^L_{\nu^\prime}(-x)=W^0\big(-x-(\zeta_{\nu^\prime}- d )\big)=W^0\big(x+(\zeta_{\nu^\prime}- d )\big)=W^0\big(x-(-\zeta_{\nu^\prime}+ d )\big)$. Hence, 
 \[ {\rm supp}(W^L_{\nu})=\{y\ :\ |y-(\zeta_{\nu}- d )|<\delta_{\lambda,M,D}\ \},\quad {\rm supp}(W^R_{\nu^\prime})=\{x\ :\ |x-(-\zeta_{\nu^\prime}+ d )|<\delta_{\lambda,M,D}\ \},\quad 
 \]
 Hence, 
 \eql{
\begin{aligned}
|\Omega^\radial(\nu^\prime,\nu)|
&\le \tau^2\,\|W^0\|_\infty^2
\int_{\substack{
  |y-(\zeta_{\nu}- d )|<\delta_{\lambda,M,D}\\
  |x-(-\zeta_{\nu^\prime}+ d )|<\delta}}
\\
&\qquad \vf^{\radial}(x- d )\,K(x-y)\,\vf^{\radial}(y+ d )\,\dif{x}\,\dif{y} .
\end{aligned}
\label{eq:Omega-nunup-1}
}
 
 We next bound $K(x-y)$. Since $x= d -\zeta_{\nu^\prime}+\mathcal{O}(\delta)$ and 
  $y=\zeta_{\nu}- d +\mathcal{O}(\delta)$, it follows that $x-y=2 d -\zeta_\nu-\zeta_{\nu^\prime}+\mathcal{O}(\delta)$. Therefore, since $\big||\zeta_\mu|-D\big|\le C$ ($\mu=1,2,3,4$), we have 
 $ |x-y|\ge D^{3/2} -2D-C$, and hence $ |x-y|^2\ge D^3 -4D^{5/2} +4D^2 -CD^{3/2}$. By \Cref{lem:bounds on the SHO Green's function}, for $D$ fixed and large enough we have 
 \eql{ 0\le K(x-y)\le C \ee^{-\frac{\lambda}{4}|x-y|^2 }\le C 
 \ee^{-\frac{\lambda}{4}\left(D^3 -4D^{5/2} +4D^2\right) +C\lambda D^{3/2}}. \label{eq:Kxy-nunup}
 }
 Further, on the region of integration in \Cref{eq:Omega-nunup-1} we have by \Cref{prop:gs0-decay}
\eql{
 0\le \vf^{\radial}(x- d )\ ,\ \vf^{\radial}(y+ d )\ \le \exp\left(-\frac{\lambda}{4}D^2 +C\lambda D\right) .
 \label{eq:phi-xy+xi}
 }
 Using \Cref{eq:Kxy-nunup} and \Cref{eq:phi-xy+xi} to bound $|\Omega^\radial(\nu,\nu^\prime)|$ in \Cref{eq:Omega-nunup-1} yields
 
 \begin{prop}\label{prop:Omeganunup-bd1}
 For $\nu, \nu'=1,2,3,4$:
 \begin{align}
   |\Omega^\radial(\nu,\nu^\prime)| &\le 
 C\tau^2 \delta^4 \ee^{ -\frac{\lambda}{4}\big[2D^2 + (D^3 -4D^{5/2} +4D^2)\big] +C\lambda D^{3/2}\ }
 \label{eq:Omegnunup-bd1}
 \end{align}
 \end{prop}

\subsubsection{ Determination of $\tau=\tau_{\lambda,M,D}$ by imposition of 
the comparison inequalities \Cref{eq:PS-PP_SS}} \label{sec:tau-determined}

With a view toward determining $\tau=\tau_{\lambda,M,D}$, we assume that the expressions $\Omega^\radial(\nu,\nu^\prime)$, are good approximations of $\Omega(\nu,\nu^\prime)$. Constraints \Cref{eq:PS-PP_SS} will be implemented as:
\begin{subequations}
\label{eq:PS-PPandSS}
\begin{align}
 \textrm{\Cref{eq:PS-PP}}\ &:\    \big|\Omega^\radial(0,\nu)\big|\ \gg \ 
 \big|\Omega^\radial(0,0)\big|,\ \nu=1,2\label{eq:PS-PP-1}\\
\textrm{\Cref{eq:PS-SS}}\ &:\ \big|\Omega^\radial(0,\nu)\big|\ \gg \Big|\Omega^\radial(\nu',\nu'')\Big|,\ \nu=1,2,\nu',\nu'' = 1,2,3,4 \label{eq:PS-SS-1}
  \end{align}
  \end{subequations}
 We'll note also that:
\[
 \min\big(\big|\Omega^\radial(0,1)\big|, \big|\Omega^\radial(0,2)\big|\big) \gg \max\big(\big|\Omega^\radial(0,3)\big|, \big|\Omega^\radial(0,4)\big|\big)
 \]
 
  
\noindent We next derive upper and lower bounds on $\tau\delta^2$, which are sufficient conditions for  \Cref{eq:PS-PPandSS} to hold.

Note from Part 1 of \Cref{prop:dominant-Omega0nu} 
%
%
and \Cref{prop:Omega00-bd} that a sufficient condition for \Cref{eq:PS-PP-1} to hold, for $D\ge D_{\rm max}$ fixed and all $\lambda\ge\lambda_{\rm max}$, is 

\[
\tau\delta^2 \times 
 \ee^{-\frac{\lambda}{4}D^2 -\frac{\lambda}{4}(D^{3/2}-D)^2}\times \ee^{-C\lambda D^{3/2}} \gg\ \lambda^4 \ee^{-\frac{\lambda}{4}D^3 + C\lambda D^{3/2}}
\]
or 
\eql{ \tau\delta^2\ \gg\ \ee^{\frac{\lambda}{2}D^2} \ee^{C\lambda D^{3/2}} \ee^{-\frac{\lambda}{2}D^{5/2}} 
\label{eq:tau_sq-lb}
}

Next note, again from Part 1 of \Cref{prop:dominant-Omega0nu} (for $\nu=1,2$)
 and \Cref{prop:Omeganunup-bd1} 
 that a sufficient condition for \Cref{eq:PS-SS-1} to hold, for $D\ge D_{\rm max}$ fixed and all $\lambda\ge\lambda_{\rm max}$, is
\[
\tau\delta^2 \times 
 \ee^{-\frac{\lambda}{4}D^2 -\frac{\lambda}{4}(D^{3/2}-D)^2}\times \ee^{-C\lambda D^{3/2}} \gg\ 
 \tau^2 \delta^4 \times \ee^{ -\frac{\lambda}{4}\big[2D^2 + (D^3 -4D^{5/2} +4D^2)\big] +C\lambda D^{3/2}} ,
\]
or equivalently
\eql{
e^{\lambda D^2} \ee^{-\frac{\lambda}{2}D^{5/2}}
 \ee^{-C\lambda D^{3/2}}\ \gg\ \delta^2\tau.  \label{eq:tau_sq-ub}
}
Combining \Cref{eq:tau_sq-lb} and \Cref{eq:tau_sq-ub} we have
\[
 \ee^{-\frac{\lambda}{2}D^{5/2}} \ee^{C\lambda D^{3/2}}\ee^{\frac{\lambda}{2}D^2} \ \ll\ 
 \delta^2\tau\ \ll\ \ee^{-\frac{\lambda}{2}D^{5/2}} \ee^{\lambda D^2} 
\ee^{-C\lambda D^{3/2}}.
\]

Recall, from \Cref{eq:delta-1}, that 
$\delta_{\lambda,M,D}=\exp(-M\lambda D^{3/2})$, where $M$ has been fixed by the constraint just below \cref{eq:exp-bound-in-Om_error}. Therefore, the above inequalities hold provided:
\eql{
 \ee^{2M\lambda D^{3/2}} \ee^{C\lambda D^{3/2}}\ee^{-\frac{\lambda}{2}D^{5/2}} \ee^{\frac{\lambda}{2}D^2} \ \ll\ 
 \tau\ \ll\ \ee^{2M\lambda D^{3/2}}\ \ee^{-\frac{\lambda}{2}D^{5/2}} \ee^{\lambda D^2} 
 \ee^{-C\lambda D^{3/2}}\,. \label{eq:tau-determine}
}

Using the choice of $\tau$ introduced in \Cref{eq:tau} 
\eql{ \tau_{\lambda,M,D} = \ee^{-\frac{\lambda}{2}D^{5/2}}\ \ee^{\frac{3}{4}\lambda D^2} = \ee^{-\frac{\lambda}{4}(2D^{5/2}-3D^2)}\qquad . 
\label{eq:tau-chosen}
}
we see that \cref{eq:tau-determine} is now satisfied.
Substituting the value of $\delta_{\lambda,M,D}$ from \Cref{eq:delta-1}, we find: 
\eql{ \delta^2 \tau_{\lambda,M,D}= \ee^{-2M\lambda D^{3/2}}\ \ee^{-\frac{\lambda}{2}D^{5/2}}\ \ee^{\frac{3}{4}\lambda D^2}\qquad .
\label{eq:delta-sq-tau}}
The relation \cref{eq:delta-sq-tau} for $\delta^2\tau$ sets the size of the quantities: $\Omega^\radial(0,0)$,
 $\Omega^{\radial}(0,\nu)$ and 
 $\Omega^{\radial}(\nu,0)$ for $\nu=1,2,3,4$, and 
 $\Omega^{\radial}(\nu,\nu')$ for $\nu=1,2,3,4$, with the dominant terms being:
  $\Omega^{\radial}(0,\nu)$ and 
 $\Omega^{\radial}(\nu,0)$ for $\nu=1,2$.

By the definition of $\rho^{\radial}(\lambda)$, \cref{eq:rho-0-expand}, 
 \begin{align}
&\Big|-\rho^{\radial}(\lambda)\ - \sum_{\nu=1}^2\Big(\ \Omega^\radial(\nu,0)\ +\ \Omega^\radial(0,\nu)\ \Big)\Big| \nonumber\\
&\quad =\ \ \Bigg|\Omega^\radial(0,0) \ +\ 
\sum_{\nu=3}^4\Big(\ \Omega^\radial(\nu,0)\ +\ \Omega^\radial(0,\nu)\ \Big)
+ \sum_{\nu,\nu^\prime=1}^4\Omega^\radial(\nu,\nu^\prime)\Bigg| .\label{eq:rho-0-expand-4soph}
\end{align}

We may now use \cref{prop:dominant-Omega0nu} (Part 1, \cref{eq:Omega00nu-exp}) to evaluate the sum on the left hand side of \cref{eq:rho-0-expand-4soph}. Further, we can use \cref{prop:Omega00-bd} to bound $|\Omega^{0)}(0,0)|$, \cref{prop:dominant-Omega0nu} (Part 2, \cref{eq:Omega00nu-exp-12}) to bound the first sum on the right hand side and finally \cref{prop:Omeganunup-bd1} to bound the second sum on the right hand side of \cref{eq:rho-0-expand-4soph}.

First, using the expression in \Cref{eq:Omega00nu-exp-12} and $ d \wedge\zeta_\nu=\pm D^{3/2}\bar{y}/2$ we directly obtain 
\begin{prop}\label{prop:Omega0-prinz}
 \begin{align}
 &\sum_{\nu=1}^2\Big(\ \Omega^\radial(\nu,0)\ +\ \Omega^\radial(0,\nu)\ \Big) \nonumber\\
&\quad =  \fint_{B_\delta(0)}|W^0|\times \ee^{-\frac{\lambda}{4}D^3}
\ \ee^{-2M\lambda D^{3/2}+\frac{1}{4}\lambda D^2}  \nonumber \\ &\qquad \times \left[\cos\left(\frac{\lambda}{2}D^{3/2}\ \bar{y} \right)+{\rm ERR}_\lambda\right] \times R(D,\lambda),
\label{eq:sum-nu-0} \end{align}
 where
 \begin{align}
   &\exp\left(-C\lambda D^{3/2}\right)\ \le\ 
   R(D,\lambda)\ \le \exp\left(C\lambda D^{3/2}\right),
  \label{eq:RDlam-bounds-34-A}
   \end{align}
   and 
   \begin{align}
& \big|{\rm ERR}_\nu\big| \le \exp\left(-\frac{M}{3}\lambda D^{3/2}\right)\ . \label{eq:ERR_nu-bd-A}
 \end{align}
\end{prop}

We note that $\Omega^\circ(\nu,0)=\overline{\Omega^\circ(0,\nu)}$

Further \cref{prop:Omega00-bd}, \cref{prop:dominant-Omega0nu} (Part 2, \cref{eq:Omega00nu-exp-12}), \cref{prop:Omeganunup-bd1}, together with 
the expression for $\delta^2\tau$ in \cref{eq:delta-sq-tau}, yield the following bound for the right hand side of \cref{eq:rho-0-expand-4soph}:
 \begin{align} 
 &\Bigg|\Omega^\radial(0,0) \ +\ 
\sum_{\nu=3}^4\Big(\ \Omega^\radial(\nu,0)\ +\ \Omega^\radial(0,\nu)\ \Big)
+ \sum_{\nu,\nu^\prime=1}^4\Omega^\radial(\nu,\nu^\prime)\Bigg| \nonumber \\
&\qquad \le 
e^{-\frac{\lambda}{4}D^3}
 \Bigg(\ \ee^{C\lambda D^{3/2}} + \left(\tau\delta^2\right)\ \ee^{-\frac{\lambda}{2}D^2} \ee^{-\frac{\lambda}{2}D^{5/2}}
 \left[\sum_{\nu=3,4}R_\nu (1+ERR_\nu)\right] + \left(\tau\delta^2\right)^2\ \ee^{-\frac{3\lambda}{2}D^2}e^{C\lambda D^{3/2}} \ee^{\lambda D^{5/2}}
  \ \Bigg) \nonumber\\
  &\qquad \lesssim 
e^{-\frac{\lambda}{4}D^3}
 \Bigg(\ \ee^{C\lambda D^{3/2}} + \ee^{-\lambda(2M-C)D^{3/2}} \ee^{-\lambda D^{5/2}} \ee^{\frac{\lambda}{4}D^2}\ +\ \ee^{-(4M-C)\lambda D^{3/2}}
  \ \Bigg)\lesssim \ee^{-\frac{\lambda}{4}D^3}\ \ee^{C\lambda D^{3/2}}\ . \label{eq:rho0-error}
\end{align}
Applying \cref{prop:Omega0-prinz} and 
\cref{eq:rho0-error} to \cref{eq:rho-0-expand-4soph}, we obtain 
\begin{prop}\label{eq:rho0-expression}
  \begin{align*}
    -\rho^{\radial}(\lambda) &= 
      \fint_{B_\delta(0)}|W^0|\times \ee^{-\frac{\lambda}{4}D^3}
\ \ee^{-2M\lambda D^{3/2}+\frac{1}{4}\lambda D^2} \times R(\lambda,D)\\
&\qquad \times \Bigg( \cos\left(\frac{\lambda}{2}D^{3/2}\ \bar{y} \right) + \widetilde{\rm ERR}(\lambda,D)
\Bigg)\ ,   
  \end{align*}
  where
  \[\big|\widetilde{\rm ERR}(\lambda,D)\big| \lesssim C\exp\br{-\frac{M}{3}\lambda D^{3/2}},\ , \]
  and 
  \begin{align}
   &\exp\left(-C\lambda D^{3/2}\right)\ \le\ 
   R(\lambda,D)\ \le \exp\left(C\lambda D^{3/2}\right)\,.
  \label{eq:RDlam-bounds-34-B}
   \end{align}
\end{prop}

\subsubsection{From an expansion for $-\rho^{\radial}(\lambda)$ to an expansion for $-\rho(\lambda)$}\label{sec:from_rho0_2_rho}

The results of \Cref{sec:compare-phi_and_phi0} on bounds for the norm of $\vf_\lambda-\vf^{\radial}_{\lambda}$ imply, for example, that terms in $\Omega(\mu,\nu)-\Omega^\radial(\mu,\nu)$ all have at least one additional factor of \begin{equation*} \tau_{\lambda,M,D}= \ee^{-\frac{\lambda}{2}D^{5/2}}\ \ee^{\frac{3}{4}\lambda D^2} = \ee^{-\frac{\lambda}{2}(2D^{5/2}-3D^2)}\qquad . 
\end{equation*}
Such terms can be bounded in a manner similar to 
\Cref{eq:error-main}.

\subsection{Expansions of $E^{(\sigma)}_\lambda$, $\Delta(\lambda)$ and $\mathfrak{S}(\lambda)$;
 proof of \cref{prop:Epm1-expand-Delta-expand}}\label{sec:splitting-proof}

We solve 
 \eql{ E - e_\lambda -\left\langle \psi_\sigma, \Big( H_\lambda -e_\lambda\Big)\psi_\sigma\right\rangle
+ B_\sigma^* \Big[ \Pi^\perp_\sigma \left( H_\lambda -E \right) \Pi^\perp_\sigma \Big]^{-1} B_\sigma\ 
 =0,
\label{eq:Esigma-eqn1}
}
in a disc $|E-e_\lambda|<\varepsilon_0 c_{\rm gap}$, with $\varepsilon_0$
sufficiently small, for $E=E^{(\sigma)}_\lambda$ the $L^2_\sigma-$ ground state eigenvalue of $ h_\lambda$. 
We make use the following four claims, which hold for all $\lambda\ge\lambda_{\rm min}$ and $D\ge D_{\rm min}$:

\noindent {\bf Claim 1:} 
\begin{subequations}
 \label{eq:claim1}
\begin{align}
\|\big( H_\lambda
 -e_\lambda\Id\big)\vf_\lambda^L\| &= \| v^R_\lambda \vf_\lambda^L\| \le Ce^{-\frac{\lambda}{4}\big( D^{3/2}-2D\big)^2
 + C\lambda D^{3/2}}\\
 \|\big( H_\lambda
 -e_\lambda\Id\big)\vf_\lambda^R\| &= \| v^L_\lambda \vf_\lambda^R\| \le Ce^{-\frac{\lambda}{4}\big( D^{3/2}-2D\big)^2
 + C\lambda D^{3/2}}
 \end{align}
 \end{subequations}

\noindent {\bf Claim 2:} 
\[
\left\langle \psi_\sigma, \Big( H_\lambda -e_\lambda\Id\Big)\psi_\sigma\right\rangle
 = \frac{2\sigma\rho(\lambda)}{\|\vf_\lambda^L+\sigma\vf_\lambda^R\|^2}\ +\ 
 {\rm Error}^{(\sigma)}_\lambda,
\]
where $\rho(\lambda)$ denotes the magnetic hopping coefficient \cref{eq:hopping_coefficient} and 
\[
\Big| {\rm Error}^{(\sigma)}_\lambda\Big|
\le C \exp\big(-\frac{\lambda}{2}(D^{3/2}-2D)^2+C\lambda D^{3/2}\big).
\]

 \noindent {\bf Claim 3:} There exists a constant $C>0$, which depends on $c_{\rm gap}$, such that for all $E$ satisfying $|E-e_\lambda|\le c_{\rm gap}/2$:
 \[
\Bigg\|\ \Big[ \Pi^\perp_\sigma \left( H_\lambda -E\Id \right) \Pi^\perp_\sigma \Big]^{-1}\Big\|_{_{L^2_{\sigma,\perp}\to L^2_{\sigma,\perp}}} \le C\,.
\] 

 \noindent {\bf Claim 4:}

 \[ \|B_\sigma\|_{_{\mathbb C\to L^2_{\sigma,\perp}}}
\le C\exp\Big(-\frac{\lambda}{4}\big( D^{3/2}-2D\big)^2
 + C\lambda D^{3/2}\Big)
 \] 
and
 \[ \|B^*_\sigma\|_{_{L^2_{\sigma,\perp}\to\mathbb C}}
\le C\exp\Big(-\frac{\lambda}{4}\big( D^{3/2}-2D\big)^2
 + C\lambda D^{3/2}\Big)
 \]

\noindent We now use Claims 1-4 to prove \cref{prop:Epm1-expand-Delta-expand} and then prove these claims afterward.

Claims 1-4 imply that for any fixed $\lambda\ge\lambda_{\rm min}$, $D\ge D_{\rm min}$, and all $E$ varying over $|E-e_\lambda|\le c_{\rm gap}/2$, equation \cref{eq:Esigma-eqn1} may be expressed as:
\begin{align}\label{eq:E-eqn}
E - e_\lambda - \frac{2\sigma\rho(\lambda)}{\|\vf_\lambda^L+\sigma\vf_\lambda^R\|^2}\ +\ \mathfrak{E}_\lambda(E) = 0, 
\end{align}
where $e_\lambda\sim -\lambda^2$ (by (V4) and \cref{prop:gs-bounds}) and $|\rho(\lambda)|\lesssim \ee^{-\frac{\lambda}{4}(D^3-D^2)}$
(\cref{prop:hopping_leading+error}) and where $E\mapsto \mathfrak{E}_\lambda(E)$ is analytic for $|E-e_\lambda|= c_{\rm gap}/2$ (\cref{prop:spectralgap}) and satisfies the bound
\[ |\mathfrak{E}_\lambda(E)|\le C\exp\Big(-\frac{\lambda}{2}\big( D^{3/2}-2D\big)^2
 + C\lambda D^{3/2}\Big). \]
 By Rouch\'e's Theorem, for $\lambda\ge \lambda_{\rm min}$, and $D\ge D_{\rm min}$, 
 \cref{eq:E-eqn} has a unique solution, $E^{(\sigma)}_\lambda$, in the disc $|E-e_\lambda|< c_{\rm gap}/2$. Moreover, $E^{(\sigma)}_\lambda$ is an $L^2_\sigma$-eigenvalue, and is therefore real by self-adjointness.
 Finally, applying Claims 1-4 to \cref{eq:E-eqn}, with $E=E^{(\sigma)}_\lambda$, as well as the results just obtained on the hopping coefficient, we obtain all conclusions of \cref{prop:Epm1-expand-Delta-expand}.

It remains now to prove Claims 1-4.

 \subsubsection{Proof of Claim 1} We bound $\big\|\Big( H_\lambda - e_\lambda \Big)\big( \vf_\lambda^L+\sigma\vf_\lambda^R\big)\big\|$. By \cref{eq:HphiLR}, 
 \[
 \Big( H_\lambda - e_\lambda \Big)\big( \vf_\lambda^L+\sigma\vf_\lambda^R\big) = v^R_\lambda \vf_\lambda^L+\sigma v^L_\lambda \vf_\lambda^R.
 \]
 Hence, 
 \[ \big\|\Big( H_\lambda - e_\lambda \Big)\big( \vf_\lambda^L+\sigma\vf_\lambda^R\big)\big\|\le
 \big\| v^R_\lambda \vf_\lambda^L\big\|\ +\ \big\| v^L_\lambda \vf_\lambda^R\big\|.
 \]
 Both terms on the right hand side are bounded in the same way. Hence, we focus on bounding $\big\| v^R_\lambda \vf_\lambda^L\big\|$:
 \eq{
 \big\| v^R_\lambda \vf_\lambda^L\big\| &\le \Big(\int | v^R_\lambda(x)|^2\ |\vf_\lambda^L(x)|^2\ \dif{x}\Big)^{1/2}   \\ 
 &=\Big(\int_{{\rm supp}\ ( v_\lambda)} | v_\lambda(x')|^2\ |\vf(x'+2 d )|^2\ \dif{x}\Big)^{1/2}\\
 &\le \Big(\int_{|x'|\le CD} | v_\lambda(x')|^2\ |\vf(x'+2 d )|^2\ \dif{x}\Big)^{1/2}
 }
 Note that, by fixing $D$ large enough, we have
$|x'+2 d |=|x'+(D^{3/2},0)|\ge D^{3/2}-CD\ge D\ge C_1$. We can therefore apply the Gaussian decay bound of \cref{prop:Gauss-bound-phi_lambda} to obtain
\begin{align*}
\big\| v^R_\lambda \vf_\lambda^L\big\| &\le
C \lambda^2 \sup_{|x'|\le cD}\ 
\exp\left(-\frac{\lambda}{4}\Bigg[{\rm dist}\big((x'_1+D^{3/2},0), B_{CD}(0)\big)\ \Bigg]^2\right)\\
&\le C \lambda^2 
\exp\left(-\frac{\lambda}{4}(D^{3/2}-2CD)^2\right),
\end{align*}
which implies \cref{eq:claim1}.
This completes the proof of Claim 1.

  \subsubsection{Proof of Claim 2} We expand the inner product $\left\langle \psi_\sigma, \Big( H_\lambda -e_\lambda\Big)\psi_\sigma\right\rangle$. Writing out the terms explicitly gives:
  \begin{align*}
  \left\langle \psi_\sigma, \big( H_\lambda -e_\lambda\big)\psi_\sigma\right\rangle &= 
  \frac{1}{\|\vf_\lambda^L+\sigma \vf_\lambda^R\|^2}
  \left\langle \big( H_\lambda -e_\lambda\big)(\vf_\lambda^L+\sigma \vf_\lambda^R),(\vf_\lambda^L+\sigma \vf_\lambda^R) \right\rangle\\
  &= 
  \frac{1}{\|\vf_\lambda^L+\sigma \vf_\lambda^R\|^2}
  \Big[\ \sigma\ \left\langle \big( H_\lambda -e_\lambda\big)\vf_\lambda^L,\vf_\lambda^R \right\rangle\ +\ \sigma\ \left\langle \big( H_\lambda -e_\lambda\big)\vf_\lambda^R,\vf_\lambda^L \right\rangle \Big]\\
  &\quad +\ 
  \frac{1}{\|\vf_\lambda^L+\sigma \vf_\lambda^R\|^2}
  \Big[ \left\langle \big( H_\lambda -e_\lambda\big)\vf_\lambda^L,\vf_\lambda^L \right\rangle\ +\ \left\langle \big( H_\lambda -e_\lambda\big)\vf_\lambda^R,\vf_\lambda^R \right\rangle \Big]
  \end{align*}
  Applying \cref{eq:HphiLR} yields
\begin{align*}
 \left\langle \psi_\sigma, \big( H_\lambda -e_\lambda\big)\psi_\sigma\right\rangle = &\frac{1}{\|\vf_\lambda^L+\sigma \vf_\lambda^R\|^2}
  \Big[\ \sigma\ \left\langle v^R_\lambda \vf_\lambda^L,\vf_\lambda^R \right\rangle\ +\ \sigma\ \left\langle v^L_\lambda \vf_\lambda^R,\vf_\lambda^L \right\rangle \Big]\\
  &\quad +\ 
  \frac{1}{\|\vf_\lambda^L+\sigma \vf_\lambda^R\|^2}
  \Big[ \left\langle v^R_\lambda \vf_\lambda^L,\vf_\lambda^L \right\rangle\ +\ \left\langle v^L_\lambda \vf_\lambda^R,\vf_\lambda^R \right\rangle \Big]\\
  &= \frac{1}{\|\vf_\lambda^L+\sigma \vf_\lambda^R\|^2}
  \Big[ 2\sigma \rho(\lambda) +\left\langle v^R_\lambda \vf_\lambda^L,\vf_\lambda^L \right\rangle\ +\ \left\langle v^L_\lambda \vf_\lambda^R,\vf_\lambda^R \right\rangle\ .
  \Big]
  \end{align*}
  Note that the denominator is easily seen to be bounded below. The latter two terms are bounded above by applying the Gaussian bound of \cref{prop:Gauss-bound-phi_lambda} to 
  $\big\| v^R_\lambda \ |\vf_\lambda^L|^2\ \big\|$ and 
  $\big\| v^L_\lambda \ |\vf_\lambda^R|^2\ \big\|$, analogous to estimates in the proof of Claim 1. This gives an upper bound for the latter two terms:
  $\sim C \lambda^2 
\exp\left(-\frac{\lambda}{2}(D^{3/2}-2CD)^2\right)$, and the proof of Claim 2 is now complete.

  \subsubsection{Proof of Claim 3}
  The proof of this claim is part of \cref{prop:spectralgap}.

   \subsubsection{Proof of Claim 4} We derive a bound on the norm of $B_\sigma: {\mathbb C}\to L^2_{\sigma,\perp}$ and hence also on the norm of 
   $B^*_\sigma: L^2_{\sigma,\perp}\to {\mathbb C} $.
   Since $\vf_\lambda^L+\sigma \vf_\lambda^R= \|\vf_\lambda^L+\sigma \vf_\lambda^R\|\ \psi_\sigma$, we have from \cref{eq:psiperp-alpha-sys} that 
  \begin{align*}
\Big( H_\lambda-e_\lambda\Big)\left(\vf_\lambda^L+\sigma \vf_\lambda^R\right) &=
 \|\vf_\lambda^L+\sigma \vf_\lambda^R\|\ \Big( H_\lambda-e_\lambda\Big)\psi_\sigma\\
 &= \|\vf_\lambda^L+\sigma \vf_\lambda^R\|\ \Big( \ \left\langle \psi_\sigma, \big( H_\lambda -e_\lambda\big)\psi_\sigma\ \right\rangle\ \psi_\sigma + B_\sigma \Big)\ .
  \end{align*}
  Here,
  $B_\sigma = \Pi_\sigma^\perp\left( H_\lambda\psi_\sigma\right) $ and the two terms on the right hand side are orthogonal.
  Since $\|B_\sigma\|^2_{\mathbb C\to L^2_\sigma}=\|B_\sigma\|^2_{L^2_\sigma}$,  the Pythagorean theorem implies 
  \[ 
  \|B_\sigma\|_{\mathbb C\to L^2_\sigma} \le
  \frac{1}{\|\vf_\lambda^L+\sigma \vf_\lambda^R\|} 
  \ \big\|\Big( H_\lambda-e_\lambda\Big)\left(\vf_\lambda^L+\sigma \vf_\lambda^R\right)\big\|. \]
  The proof of Claim 4 is now obtained by applying Claim 1.
  
\bigskip\bigskip

\appendix
\section{The harmonic oscillator and the Landau Hamiltonian; Resolvent kernels }\label{sec:HO&Landau}
\subsection{The heat kernels}
We recall that on $L^2(\RR)$, the harmonic oscillator $P^2+X^2$ has the well known Mehler kernel given, for all $x,y\in\RR,t>0$, as
\eql{
  \exp\br{-t\br{P^2+X^2}}(x,y) = \frac{1}{\sqrt{2\pi\sinh\left(2t\right)}}\exp\left(-\frac12\coth\left(2t\right)\left(x^{2}+y^{2}\right)+\frac{1}{\sinh\left(2t\right)}xy\right)\,.
} 

With this formula we can immediately write the heat kernel for the \emph{two-dimensional} harmonic oscillator, with $\omega>0$: \eql{
H^{\mathrm{SHO}}_\omega := P^2+\frac14\omega^2X^2
} on $L^2(\RR^2)$. We have $\sigma(H^{\mathrm{SHO}}_\omega)=\sigma_{\mathrm{disc.}}(H^{\mathrm{SHO}}_\omega)=\omega\br{\NN_{\geq0}+1}$ \cite[Section 8.3]{Teschl2014}. For all $t>0,x,y\in\RR^2$ we have
\eql{
\exp\left(-tH^{\mathrm{SHO}}_\omega\right)\left(x,y\right)
=\frac{\omega}{4\pi\sinh\left(\omega t\right)}\exp\left(-\frac{\omega}{4\sinh\left(\omega t\right)}\left(\cosh\left(\omega t\right)\left(\norm{x}^{2}+\norm{y}^{2}\right)-2x\cdot y \right)\right)\,.
}

The Landau Hamiltonian, acting on $L^2(\RR^2)$, in its symmetric gauge, is given by
\eq{
  H^{\mathrm{Landau}}_b\equiv\calP_b^2 \equiv (P-\frac12 b X^\perp)^2 \ ,
} with $b>0$. We have $\sigma(H^{\mathrm{Landau}}_b)=\sigma_{\mathrm{ess.}}(H^{\mathrm{Landau}}_b) = b\br{2\NN_{\geq0}+1}$ \cite[Section 110]{Landau_Lifshitz_vol_3}. 
%
%

The resolvent of the Landau Hamiltonian and that of the simple harmonic oscillator are related as follows
\eql{\label{eq:relation between Landau and SHO Green's functions}
\br{H^{\rm Landau}-z\Id}^{-1}(x,y)  = \ee^{-\ii\frac{b}{2}x\wedge y}\  \br{H^{\rm SHO}-z\Id}^{-1}(x-y,0)\,,\\\nonumber\qquad
\textrm{where}\quad z\in\br{\sigma(H^{\rm Landau}_b)\cup \sigma(H^{\rm SHO}_b)}^c\quad {\rm and }\quad x,y\in\RR^2:x\neq y\,.
}

To establish \cref{eq:relation between Landau and SHO Green's functions}, first introduce $K_{b,z}^{\mathrm{SHO}}$, the fundamental solution of $H^{\mathrm{SHO}}$ with spectral parameter $z$ and pole at $x=0$, i.e., \eq{\br{H^{\mathrm{SHO}}-z \Id}K_{b,z}^{\mathrm{SHO}}=\delta\,.} Next, we note the expansion 
\eq{
H^{\mathrm{Landau}}_b = H^{\mathrm{SHO}}_b- b L\ ,}
where $L\equiv X\wedge P$ is the angular momentum operator, 
and that $x\mapsto K_{b,z}^{\mathrm{SHO}}$ is a radial function, and so is an eigenfunction of $L$ of eigenvalue zero. We thus find, using the fact that $[H^{\rm Landau},\hat{R}^y]=0$ and 
\eq{
   \br{H^{\rm Landau}_b-z\Id} K_{b,z}^{\mathrm{SHO}} &= \delta\,,\quad\textrm{that }\\
  \br{H^{\rm Landau}_b-z\Id}\hat{R}^yK_{b,z}^{\mathrm{SHO}} &= \hat{R}^y\delta = \ee^{-\ii\frac{\lambda}{2}(\cdot \wedge y)}\delta(\cdot-y) =\delta(\cdot-y)\,.
} In other words,  $\br{H^{\rm Landau}_b-z\Id}^{-1}(x,y)=\br{\hat{R}^y K_{b,z}^{\mathrm{SHO}}}(x)$ which is precisely \cref{eq:relation between Landau and SHO Green's functions}.


A consequence of the above discussion is that:\\
{\it Bounds for the Landau resolvent kernel and Landau heat kernel can be derived from the corresponding SHO bounds.

\subsection{Some basic estimates and identities}
Above we have invoked some basic estimates involving the resolvents of the Hamiltonians discussed. We present these estimates here.
\begin{lem}[Gaussian decay of resolvent kernels]\label{lem:bounds on the SHO Green's function} Let  the spectral parameter $z\in\RR$ be such that $z<\omega$ and $x\in\RR^2\setminus\Set{0}$.  Then,  
\begin{enumerate}
\item Denote by 
$K_{\omega,z}^{\mathrm{SHO}}(x)$ the fundamental solution of $H_\omega^{\mathrm{SHO}}-z\Id$ with pole at $x=0$. Then, 
\begin{align}\label{eq:SHO-GaussBound}
  \frac{\br{1+\frac12\omega\norm{x}^2}^{-\frac12\br{1-\frac{z}{\omega}}}}{2\ee}&\leq\pi\br{1-\frac{z}{\omega}}\abs{K_{\omega,z}^{\mathrm{SHO}}(x)}\exp\br{\frac14\omega\norm{x}^2} \nonumber\\
  &\qquad\qquad\qquad\leq \ee^{\frac{z_+}{\omega}}\br{
\frac{1-\frac{z}{\omega}}{\,\omega\,\|x\|^{2}}
+\frac{1}{2(\ee-\ee^{-1})}\, ,}
\end{align} where $z_+\equiv\max\br{\Set{0,z}}$. 
\item  Let  $\br{H^{\rm Landau}_\omega-z\Id}^{-1}(x,y)$ denote the resolvent integral kernel for the Landau Hamiltonian $H^{\mathrm{Landau}}_\omega$. Since $|\br{H^{\rm Landau}_\omega-z\Id}^{-1}(x,y)|= \abs{K_{\omega,z}^{\mathrm{SHO}}(x-y)}$, the resolvent kernel upper and lower bounds \cref{eq:SHO-GaussBound} apply as well to $|\br{H^{\rm Landau}_\omega-z\Id}^{-1}(x,y)|$, where $x-y$ is substituted for $x$. 
\item The \emph{the upper bounds} of parts 1. and 2., for $\abs{K_{\omega,z}^{\mathrm{SHO}}(x-y)}$ and $|\br{H^{\rm Landau}_\omega-z\Id}^{-1}(x,y)|$, continue to hold for $z\in\CC$, with $z$ replaced by $\Re{z}$, provided $\Re{z}<\omega$.
  \end{enumerate}
\end{lem}
\begin{proof}
  Using the Laplace representation valid for $\Re{z}<\inf\sigma(H^{\mathrm{SHO}}_\omega)$ we have
  \eq{
  \br{H^{\mathrm{SHO}}_\omega-z\Id}^{-1}(x,y) = \int_{t=0}^\infty\ee^{tz}\exp\left(-tH^{\mathrm{SHO}}_\omega\right)\left(x,y\right)\dif{t}\,.
  } 
  We are only interested in the absolute value of the Green's function, which is positive if $z\in\RR$, so let us (for the sake of brevity) WLOG assume $z\in\RR$ in the sequel. Since $\sigma(H^{\mathrm{SHO}}_\omega)=\omega(\NN_{\geq0}+1)$ we find for all $z<\omega$ and $x\in\RR^2\setminus\Set{0}$
  \eq{
  \br{H^{\mathrm{SHO}}_\omega-z\Id}^{-1}(x,0) &= \frac{\omega}{4\pi}\int_{t=0}^\infty\frac{\exp\left(tz-\frac{\omega }{4}\coth\left(\omega t\right)\norm{x}^{2}\right)}{\sinh\left(\omega t\right)}\dif{t} \\
  &= \frac{1}{4\pi}\int_{s=0}^{\infty}\frac{\exp\br{-\frac{\omega}{4}\norm{x}^{2}\br{\coth(s)-\frac{4z}{\omega^{2}\norm{x}^{2}}\,s}}}{\sinh(s)}\,\dif{s}\,.
  }

  Splitting the integral into $(0,\infty)=(0,1)\cup[1,\infty)$ and using $\sinh(\xi)\geq\xi$ and $\coth(\xi)\geq\frac{1}{\xi}$ for all $\xi>0$ we get 
  \eq{
  I_{s\,\rm small} \leq \frac{\ee^{z_+/\omega}}{4\pi}\int_{s=0}^{1}\frac{\exp\left(-\frac{1 }{4}\omega\norm{x}^{2}/s\right)}{ s}\dif{s} \leq \frac{\ee^{z_+/\omega}}{\pi\omega\norm{x}^2}\exp\br{-\frac14\omega\norm{x}^2}\,;
  } for large $s$, we use $\sinh(\xi)\geq\frac12\br{1-\ee^{-2}}\ee^{\xi}$ which holds for all $\xi\geq1$ and $\coth(\xi)\geq1$, to get
  \eq{
  I_{s\,\rm large} &\equiv \frac{1}{4\pi}\int_{s=1}^{\infty}\frac{\exp\br{-\frac{\omega}{4}\norm{x}^{2}\br{\coth(s)-\frac{4z}{\omega^{2}\norm{x}^{2}}\,s}}}{\sinh(s)}\,\dif{s} \\
  &\leq \frac{1}{4\pi}\frac{2}{1-\ee^{-2}}\int_{s=1}^{\infty}\exp\br{-s-\frac{\omega}{4}\norm{x}^{2}\br{1-\frac{4z}{\omega^{2}\norm{x}^{2}}\,s}}\,\dif{s} \\
  &= \frac{1}{4\pi}\frac{2}{1-\ee^{-2}}\exp\br{-\frac\omega4\norm{x}^2}\int_{s=1}^\infty\exp\left(-\br{1-z/\omega}s\right)\dif{s} \\ 
  &= \frac{1}{2\pi}\frac{1}{1-\ee^{-2}}\exp\br{-\frac\omega4\norm{x}^2}\frac{\exp\br{-\br{1-z/\omega}}}{1-z/\omega}\,. 
  } Together this yields (for $z<\omega$ and $x\neq0$):
  \eq{
  \br{H^{\mathrm{SHO}}_\omega-z\Id}^{-1}(x,0) \leq 
\left[
\frac{1}{\,\omega\,\|x\|^{2}}
+\frac{1}{2(\ee-\ee^{-1})}\,
\frac{\omega}{\omega-z}
\right]\frac{\ee^{z_+/\omega}}{\pi}\ee^{-\frac14\omega\norm{x}^2}\,.
  }

  For the lower bound, assuming $z\in\RR$, we rewrite
  \eq{
  \br{H^{\mathrm{SHO}}_\omega-z\Id}^{-1}(x,0)
  &= \frac{1}{2\pi}\int_{s=0}^{\infty}\frac{\exp\br{-\frac{\omega}{4}\norm{x}^{2}f(s)}}{1-\ee^{-2s}}\,\dif{s}
  }
  with
  \eq{
    f(s) := \coth(s)+\br{1-z/\omega}\frac{4}{\omega\norm{x}^{2}}\,s\,.
  } We use $1-\ee^{-2s} \leq 1$ and $\coth(s)=1+\frac{2}{\ee^{2s}-1}$ to get 
  \eq{
  \br{H^{\mathrm{SHO}}_\omega-z\Id}^{-1}(x,0) \geq \frac{1}{2\pi} \exp\br{-\frac14\omega\norm{x}^2}\int_{s=0}^\infty\exp\br{-\frac{\omega}{2}\frac{\norm{x}^2}{\ee^{2s}-1}-\br{1-z/\omega}s}\dif{s}\,.
  } Now iff $s \geq \frac12\log\br{1+\frac12\omega\norm{x}^2} $ then $\frac{\omega}{2}\frac{\norm{x}^2}{\ee^{2s}-1} \leq 1$ and so we have \eq{
  \int_{s=0}^\infty\exp\br{-\frac{\omega}{2}\frac{\norm{x}^2}{\ee^{2s}-1}-\br{1-z/\omega}s}\dif{s} &\geq \ee^{-1}\int_{s=\frac12\log\br{1+\frac12\omega\norm{x}^2} }^\infty\exp\br{-\br{1-z/\omega}s}\dif{s} \\
  &\geq \ee^{-1}\frac{\exp\br{-\br{1-z/\omega}\frac12\log\br{1+\frac12\omega\norm{x}^2} }}{1-z/\omega} \\
  &= \frac{1}{\ee\br{1-z/\omega}}\br{1+\frac12\omega\norm{x}^2}^{-\frac12\br{1-z/\omega}}\,.
  }

  Combining everything together we find
  \eq{
  \br{H^{\mathrm{SHO}}_\omega-z\Id}^{-1}(x,0) \geq \frac{1}{2\pi\ee\br{1-z/\omega}}\exp\br{-\frac14\omega\norm{x}^2}\br{1+\frac12\omega\norm{x}^2}^{-\frac12\br{1-z/\omega}}
  }

  This completes the proof of \cref{lem:bounds on the SHO Green's function}.
\end{proof}

\section{Gaussian decay of ground states of magnetic Hamiltonians with compact wells}\label{app:gs0-decay}
\begin{prop}\label{prop:Gaussian decay of bound states for compactly supported potentials}
  Let $u:\RR^2\to\RR$ be bounded, smooth and compactly supported. Say $\psi,E$ is a normalized ground state eigen pair for $H:=\calP_\lambda^2+\lambda^2u(X)$ where $\calP_\lambda^2$ is the Landau Hamiltonian. If $E<\lambda$ then there is some explicit constant $C_\lambda(E,x)<\infty$ such that
  \eq{
  \abs{\psi(x)}\leq \lambda^2 C_\lambda(E,x)\norm{u}_\infty\sqrt{\abs{\supp(u)}}\exp\br{-\frac\lambda4\dist(x,\supp(u))^2}\qquad(x\in\RR^2\setminus\Set{0})\,.
  }
  Here $|\supp(u)|$ denotes the Lebesgue measure of $\supp(u)$ and $C_\lambda(E,x)$ is given by   \eql{\label{eq:Constant for Gaussian decay of bound states}C_\lambda(E,x):= 
  \frac{\ee^{\frac{E_+}{\lambda}}}{\pi\br{1-\frac{E}{\lambda}}}\br{
\frac{1-\frac{E}{\lambda}}{\,\lambda\,\|x\|^{2}}
+\frac{1}{2(\ee-\ee^{-1})}}\,.
} In particular if $\norm{x}\geq a$ for some $a>0$ and $E\leq 0$ we have
\eq{C_\lambda(E,x)\leq 
  \frac{1}{\pi}\br{
\frac{1}{\,\lambda a^2}
+\frac{1}{2(\ee-\ee^{-1})}}\leq 1
} for $\lambda$ sufficiently large.
\end{prop}
\begin{proof}
  We rewrite $H\psi\equiv E\psi$ as \eq{
  \br{\calP_\lambda^2-E\Id}\psi=-\lambda^2u(X)\psi\,.
  } Since $E<\lambda$, we may invert $\br{\calP_\lambda^2-E\Id}$. We then invoke the estimates above (combining \cref{lem:bounds on the SHO Green's function} and \cref{eq:relation between Landau and SHO Green's functions}) on the Landau resolvent to obtain 
  \eq{
  |\psi(x)| &\leq \int_{y\in\RR^2}\abs{\br{\calP_\lambda^2-E\Id}^{-1}(x,y)}\lambda^2\abs{u(y)}\abs{\psi(y)}\dif{y} \\
  &\leq \lambda^2 \norm{u}_\infty\sqrt{\int_{y\in\supp(u)}\abs{\br{\calP_\lambda^2-E\Id}^{-1}(x,y)}^2\dif{y}} \\
  &\leq C_\lambda(E,x)\lambda^2 \norm{u}_\infty\sqrt{\int_{y\in\supp(u)}\exp\br{-\frac12\lambda\norm{x-y}^2}\dif{y}}
  } where $C_\lambda(E,x)$ is given by \cref{eq:Constant for Gaussian decay of bound states} (it is the constant appearing on the RHS of \cref{lem:bounds on the SHO Green's function} with $\omega\mapsto\lambda,z\mapsto E$). Now if $\dist(x,\supp(u))>0$ then for all $y\in\supp(u)$, $\norm{x-y}\geq\dist(x,\supp(u))$; otherwise use the trivial bound by zero.
\end{proof}

\section{Comparing $\vf_\lambda$ with $\vf^\circ_\lambda$}\label{sec:compare-phi_and_phi0}

Our goal in this section is to compare the ground state eigenpairs of the Hamiltonians $ h^0_\lambda$ and $ h_\lambda$.

Let us first review the setup.
 The single-well Hamiltonian is

\begin{align*} 
 h_\lambda &= \calP_\lambda^2+ \vzero(x) + \tau_{\lambda,M,D}W(x)\ =\ h^\circ_\lambda + \tau_{\lambda,M,D}W(x).
 \end{align*}
\noindent {\bf Assumptions:}
 \begin{subequations}
 \label{eq:setup-ground-states}
 \begin{align}
 & \textrm{$\vzero$ and $W$ are smooth}\label{eq:vzeroW-smooth}\\
 & \|\vzero\|_\infty, \ \|W\|_\infty \le C_1 \label{eq:vzeroW}\\
 &\vzero(x),\ W(x) \le0 \label{eq:VWle0}\\
&\vzero\textrm{ is radial},\quad {\rm supp}\ \vzero \subset B_{r_0}(0) \label{eq:suppvzero}\\
&{\rm supp}\ W \subset \RR^2\setminus B_{r_0+D}(0);\ D\ge C_2,\ \textrm{dependent on $c_1$ and $C_1$} \label{eq:suppW}
\\
&0<\tau< \tau_0, \textrm{a small constant dependent only on $c_1$ and $C_1$}\label{eq:tau-bound}\\
& \textrm{The ground state eigenvalue, $e^\circ_\lambda$, of $h^\circ_\lambda$ satisfies }\nonumber \\
&\qquad\qquad -c_1\lambda^2\le e^\circ_\lambda\le C_1\lambda^2, \label{eq:gs-energy}\\
&\qquad\qquad \textrm{dist}\Big(e^\circ_\lambda, \sigma(h_\lambda^\circ)\setminus\{e^\circ_\lambda\}\Big)\ge c_{\rm gap} >0\\ 
&\lambda\ge C_1,\ \textrm{a large enough constant dependent on $c_1$ and $C_1$.} \label{eq:lam-big}
 \end{align} 
 \end{subequations}

We note that $D$ appearing above is actually $D-C$ from the main text, and $r_0=1$ was used in the main text.

In this appendix, constants $c,C$ etc are determined by the constants in the above assumptions.

\begin{claim}
 Under assumptions \Cref{eq:setup-ground-states}, the Hamiltonian  $h_\lambda$ has a ground state eigenpair 
 $(\vf_\lambda,e_\lambda)$: 
 \begin{align}
 h_\lambda\  \vf_\lambda &= e_\lambda \vf_\lambda,\quad \|\vf_\lambda\|_{L^2}=1,
\label{eq:evp-single-well} \end{align}
 with $e_\lambda\sim e_\lambda^\circ$
 and $\vf_\lambda\sim \vf_\lambda^\circ$.
\end{claim}
  \subsubsection*{Bounds on: $|e_\lambda-e^\circ_\lambda|$ and $\|\vf_\lambda - \vf^\circ_\lambda\|_{L^2}$}
Using that $ v_\lambda$ is a small perturbation of $\lambda^2v^{\radial}$, we prove the following:
 \begin{prop}\label{prop:gs-bounds} In the setting of assumptions \Cref{eq:setup-ground-states},
   \begin{align}
     |e_\lambda-e^\circ_\lambda| &\le \|W\|_\infty\ \tau_{\lambda,M,D}, \label{eq:ediff}\\
     |\vf^\circ_\lambda(x)| &\le C\lambda^2\ \exp\left(-\frac{\lambda}{4}\big[{\rm dist}\big(x,\supp \vzero \big)\big]^2\right),\quad \textrm{if $x\in{\rm supp}(W)$} \label{eq:phi0-bd-on-suppW-1}\\
     \|\vf_\lambda - \vf^\circ_\lambda\|_{L^2}
 &\le C\ |\tau|\ \lambda^2\ \exp\left(-\frac{\lambda}{4}D^2\right). \label{eq:phi-diff-L2}
   \end{align}
 \end{prop}
\noindent The remainder of this section is devoted to the proof of \Cref{prop:gs-bounds}. 

We express $\vf_\lambda$ in the form 
 \eql{ \vf_\lambda = \beta\vf^\circ_\lambda + u,\quad \textrm{where $\beta\in\CC$ and $u\perp\vf^\circ_\lambda$.}
 \label{eq:phi-expand}}
 Without any loss of generality, we may assume that $\beta$ is real and non-negative.
 We next estimate $e_\lambda-e^\circ_\lambda$ and $u$.
\begin{claim}\label{claim:e-diff}
   $|e_\lambda-e^\circ_\lambda|\le \|W\|_\infty\ \tau$.
\end{claim}
\begin{proof}
  Note that $ h_\lambda\le h_0^\lambda$, since $W\le0$. Hence, $e_\lambda\le e^\circ_\lambda$. On the other hand,
  \[ e_\lambda=\left\langle h_\lambda\vf_\lambda,\vf_\lambda\right\rangle = \left\langle h^0_\lambda\vf_\lambda,\vf_\lambda\right\rangle + \left\langle \tau_{\lambda,M,D}W\vf_\lambda,\vf_\lambda\right\rangle\ge e^\circ_\lambda - \tau\|W\|_\infty. \]
  Therefore, $- \tau\|W\|_\infty\le e_\lambda-e^\circ_\lambda\le0$. This implies \cref{eq:ediff}.
\end{proof}

We next estimate $u$, defined in \cref{eq:phi-expand}. Note $L^2(\RR^2)= \CC\vf^\circ_\lambda \oplus \{\vf^\circ_\lambda\}^\perp$ and introduce the orthogonal projection operators $\Pi^\circ=\Pi^\circ_\lambda:L^2(\RR^2) \to \CC\vf^\circ_\lambda $ and $\Pi^{\circ,\perp}:L^2(\RR^2) \to \{\vf^\circ_\lambda\}^\perp $. Next let
\eql{
\mathcal{G}_0 = \Pi^{\circ,\perp}\Big( h^0_\lambda-e^\circ_\lambda\Id \Big)^{-1} \Pi^{\circ,\perp} 
\label{eq:G0-def}}
By the spectral gap \Cref{assume:gap} we have that
\eql{\textrm{$\|\mathcal{G}_0\|_{\mathcal{B}(L^2)}\le C$, where $C$ is independent of $\lambda$ for $\lambda\ge C_1$ .}
\label{eq:CG0-bound}
}

By \Cref{eq:phi-expand}, $ (h_0^\lambda + \tau_{\lambda,M,D}W - e_\lambda)(\beta\vf^\circ_\lambda+u)=0$, and since $h_0^\lambda\vf^\circ_\lambda=e^\circ_\lambda\vf^\circ_\lambda$,  it follows that
\begin{align}\label{eq:f-def}
 f\equiv (h_0^\lambda - e^\circ_\lambda) u = -\beta\tau_{\lambda,M,D}W\vf^\circ_\lambda -\tau_{\lambda,M,D}W u + (e_\lambda-e^\circ_\lambda)\beta \vf^\circ_\lambda + (e_\lambda-e^\circ_\lambda)u.
\end{align}  
Then, since $u\perp \vf^\circ_\lambda$, 
we have 
\begin{equation*}
  f=\Pi^{\circ,\perp} f\ \textrm{and}\ u=\mathcal{G}_0f.
\end{equation*}
Then,
 \begin{align*}
 f = -\beta\tau_{\lambda,M,D}W\vf^\circ_\lambda -\tau_{\lambda,M,D}W \mathcal{G}_0f + (e_\lambda-e^\circ_\lambda)\beta \vf^\circ_\lambda + (e_\lambda-e^\circ_\lambda)\mathcal{G}_0f,
\end{align*} 
and hence by applying $\Pi^{\circ,\perp}$:
 \begin{align*}
 f = -\beta\tau_{\lambda,M,D}\Pi^{\circ,\perp} W\vf^\circ_\lambda -\tau_{\lambda,M,D}\Pi^{\circ,\perp} W \mathcal{G}_0f + (e_\lambda-e^\circ_\lambda)\Pi^{\circ,\perp} \mathcal{G}_0f.
\end{align*} 
Thus,
\[
\Big(\ I + \tau_{\lambda,M,D}\Pi^{\circ,\perp} W \mathcal{G}_0 + (e^\circ_\lambda-e_\lambda)\Pi^{\circ,\perp} \mathcal{G}_0
\Big)f = -\beta\tau_{\lambda,M,D}\Pi^{\circ,\perp} W\vf^\circ_\lambda.
\]
Since $|e^\circ_\lambda-e_\lambda|\le C\tau$ (\Cref{claim:e-diff}), $\|\mathcal{G}_0\|_{L^2\to L^2}\le C$ (bound \Cref{eq:CG0-bound}) and $\|W\|_\infty\le C$, the operator perturbing the identity has norm bounded by $\lesssim\tau$. Hence, the Neumann series converges in $\mathcal{B}(L^2)$, and we have 
\eql{\label{eq:fL2bound}
  \|f\|_{L^2}\le C\ |\beta|\ |\tau| \ \Big\|W\vf^\circ_\lambda\Big\|_{L^2} \le
  C\ |\tau| \ \Big\|W\vf^\circ_\lambda\Big\|_{L^2}.
}
The latter inequality in \cref{eq:fL2bound} follows since 
\eql{\label{eq:beta_le1}
0\le \beta\le 1.}
Indeed, since $\vf^\circ_\lambda$ and $\vf_\lambda$ are normalized, the decomposition \Cref{eq:phi-expand} implies
 $1 = \beta^2 + \|u\|^2$, where we recall that $\beta$ could be taken to be real and non-negative. This implies \cref{eq:beta_le1}.
 
To bound $\|u\|_{L^2}$, recall from \cref{eq:f-def} that $u=\mathcal{G}_0f$ and $\|\mathcal{G}_0\|_{L^2\to L^2}\le C$. 
Therefore, by \Cref{eq:fL2bound}
\eql{\label{eq:u-bound}
  \|u\|_{L^2} \le C\ \|f\|_{L^2} \le C |\tau|\ \Big\|W\vf^\circ_\lambda\Big\|_{L^2}.
}

Hence, we next bound the $L^2$ norm of $W\vf^\circ_\lambda$. For this we apply \Cref{prop:gs0-decay} to obtain a pointwise upper bound on $\vf^\circ_\lambda$ on the support of $W$.
Note that $D\ge C_1$, the constant appearing in \Cref{prop:gs0-decay}. 
Then we have 
\eql{
|\vf^\circ_\lambda(x)|
\le C\lambda^2\ \exp\left(-\frac{\lambda}{4}[{\rm dist}(x,\supp(\vzero))]^2\right),\ \  {\rm dist}(x,\supp(\vzero))\ge D.
\label{eq:gs0-decay1}}

Therefore, 
\eql{
|\vf^\circ_\lambda(x)| \le C\lambda^2\ 
\exp\left(-\frac{\lambda}{4}D^2\right),\quad \textrm{if $x\in{\rm supp}(W)$.}
\label{eq:phi0-bd-on-suppW}
}

Bounding the $L^2$ norm of $W\vf^\circ_\lambda$ using \Cref{eq:phi0-bd-on-suppW} we obtain
\eql{\label{eq:L2norm_of_Wphi0}
  \|W\vf^\circ_\lambda\|_{L^2}\le 
  C\lambda^2\ \sqrt{\Big|{\rm supp}(W)\Big|}\ 
\exp\left(-\frac{\lambda}{4}D^2\right) .
}
Therefore, substituting \Cref{eq:L2norm_of_Wphi0} into our earlier bound on $u$, \Cref{eq:u-bound}, yields
\eql{
  \label{eq:u-bound2}
 \|u\|_{L^2} \le  C\lambda^2\ |\tau|\ \sqrt{\Big|{\rm supp}(W)\Big|}\ 
\exp\left(-\frac{\lambda}{4}D^2\right) 
}
Recall that
 $1 = \beta^2 + \|u\|^2$ and 
 $0\le \beta\le 1$; see \cref{eq:beta_le1}. Hence, $1-\beta = \|u\|_{L^2}^2 / (1+\beta)$. Since $\vf_\lambda = \vf^\circ_\lambda + (\beta-1)\vf^\circ_\lambda + u$, we have 
 \[
 \|\vf_\lambda - \vf^\circ_\lambda\|_{L^2}
 \le |\beta-1| + \|u\|_{L^2}\le \|u\|_{L^2}^2 / (1+\beta) + \|u\|_{L^2}.
 \]
 Hence, for $\lambda\geq  C_1$ large enough
 \[
 \|\vf_\lambda - \vf^\circ_\lambda\|_{L^2}
 \le C\ |\tau|\ \lambda^2\ \exp\left(-\frac{\lambda}{4}D^2\right) .
 \]
 The proof of \Cref{prop:gs-bounds} is now complete.

		\pagebreak
		\begingroup
		\let\itshape\upshape
		\printbibliography
		\endgroup
	\end{document}